\documentclass[preprint,12pt]{elsarticle}

\journal{Mathematical Biosciences}
\biboptions{sort&compress}

\usepackage{amsmath,amssymb,amsthm,mathtools}
\usepackage{graphicx}
\usepackage{array}
\usepackage{booktabs}
\usepackage{float}
\usepackage{tabularx}
\usepackage{longtable}
\usepackage{tikz}
\usepackage[margin=1in]{geometry}
\usepackage[hidelinks]{hyperref}

\AtBeginDocument{%
  \hypersetup{
    pdftitle={Graph-Induced Tensor Liftings for Networked SEIR Models: Dimensional Reduction and Residual Analysis},
    pdfauthor={Enrique Baeyens},
    pdfsubject={Graph-induced tensor lifting for deterministic networked SEIR models},
    pdfkeywords={networked SEIR models, epidemic networks, graph-induced lifting, Carleman lifting, residual bounds, dimensional reduction}
  }%
}

\newtheorem{proposition}{Proposition}

\newtheorem{lemma}{Lemma}
\newtheorem{corollary}{Corollary}
\newtheorem{remark}{Remark}
\newcounter{algorithm}
\renewcommand{\thealgorithm}{\arabic{algorithm}}
\newcommand{\matA}{\mathbf{A}}
\newcommand{\matM}{\mathbf{M}}
\newcommand{\matH}{\mathbf{H}}
\newcommand{\matB}{\mathbf{B}}
\newcommand{\matK}{\mathbf{K}}
\newcommand{\Id}{\mathbf{I}}
\newcommand{\matzero}{\mathbf{0}}

\begin{document}

\begin{frontmatter}

\title{{\bfseries Graph-Induced Tensor Liftings for Networked SEIR Models: Dimensional Reduction and Residual Analysis}}

\author[uva]{Enrique Baeyens\corref{cor1}}
\cortext[cor1]{Corresponding author. Tel.: +34 983 42 39 09}
\ead{enrique.baeyens@uva.es}

\affiliation[uva]{
  organization={Department of Systems Engineering and Automatic Control and
  Institute of Advanced Production Technologies (ITAP), University of Valladolid},
  addressline={Paseo del Prado de la Magdalena, 3--5},
  postcode={47011},
  city={Valladolid},
  country={Spain}
}

\begin{abstract}
Networked SEIR models describe epidemic spread within and between interacting
subpopulations through contact-supported nonlinear transmission. Standard
polynomial liftings based on complete ordered Kronecker tensors yield linear
higher-dimensional representations, but their dimensions grow rapidly because
they retain interactions absent from the transmission graph. This paper
develops a graph-induced tensor lifting whose observables are selected from the
effective transmission support.

An exact edge-based quadratic representation separates linear compartmental
transitions from nonlinear infection terms. A homogeneous hierarchy is then
constructed recursively. The quadratic transmission field generates the next
degree. The linear compartmental field saturates the resulting dictionary
within that degree. The first edge-closure dynamics are linear up to an
explicit cubic truncation residual, and higher-order truncations contain only
next-degree terms.

The first lifted dimension scales with the numbers of subpopulations and effective transmission channels. At fixed order, graph-induced dictionaries grow linearly with network size under uniformly bounded local connectivity, whereas complete polynomial liftings retain order-dependent polynomial growth.  Uniform first edge-closure residual bounds depend on the transmission rate and the maximum weighted incoming transmission intensity. Numerical illustrations compare equal intensity per active channel with equal total incoming intensity.  They confirm that dictionary dimensions depend only on graph support, whereas residual trajectories also reflect weight accumulation, weight distribution, and nonlinear propagation. These results provide a structured basis for reduced modeling and subsequent model-specific analysis and control.

\end{abstract}

\end{frontmatter}

\section{Introduction}
\label{sec:introduction}

Compartmental epidemic models provide a classical mechanism-based description of
infectious disease dynamics. Since the work of Kermack and McKendrick,
susceptible--infectious--recovered and related formulations have supported the
analysis of outbreaks, endemic behavior, and interventions
\cite{KermackMcKendrick1927,AndersonMay1991,Hethcote2000}. Threshold quantities
and their construction for compartmental systems are treated systematically in
\cite{VanDenDriesscheWatmough2002,DiekmannHeesterbeekRoberts2010}.

Contact heterogeneity and spatial organization motivate network and
metapopulation epidemic models. In these models, nodes represent individuals,
groups, regions, or patches, while edges encode effective contacts or
movement-mediated coupling
\cite{KeelingEames2005,Newman2002,ColizzaVespignani2007}. This literature has
established that topology can shape epidemic thresholds and transients.
Broader authoritative treatments are given in
\cite{PastorSatorrasEtAl2015,KissMillerSimon2017}.

Contact-level organization is also central to pairwise and edge-based
formulations \cite{HouseKeeling2011,Volz2008,MillerSlimVolz2012}.
Effective-degree models \cite{LindquistEtAl2011}, message-passing methods
\cite{KarrerNewman2010}, and graph-automorphism lumping
\cite{SimonTaylorKiss2011} provide related descriptions or reductions using
degree classes, directed-edge messages, or symmetry classes. Together, these
approaches show that contact-supported transmission and edge-level variables
are already well established in epidemiology.

The networked SEIR model studied here belongs to this class of structured
epidemic systems. Each subpopulation contains susceptible, exposed,
infectious, and removed compartments. Transmission may occur both within a
subpopulation and between interacting subpopulations, and both mechanisms are
represented through effective transmission channels. These channels describe
cross-population infection pressure rather than explicit migration of the
compartment states, so each local population retains its own compartmental
mass. The products associated with the channels are used here differently from
classical pairwise or edge-based closures. They seed an analytical polynomial
observable hierarchy. The original deterministic node-level SEIR state remains
part of the lifted state.

In this setting, the nonlinear terms arise from susceptible--infectious
interactions along admissible transmission channels, whether within one
subpopulation or between two interacting subpopulations. They are therefore
supported only on the effective contacts through which transmission can occur,
rather than on arbitrary quadratic state interactions.

Koopman operator theory provides the broader observable-space viewpoint
\cite{Koopman1931,Mezic2005,BruntonBudisicKaiserKutz2022}. Dynamic mode
decomposition and its Koopman interpretation
\cite{RowleyEtAl2009,Schmid2010}, together with extended dynamic mode
decomposition \cite{WilliamsKevrekidisRowley2015}, show how finite sets of
observables can support linear representations and prediction. Finite lifted
predictors are developed, for example, in \cite{KordaMezic2018}. Prescribed or
learned finite dictionaries are therefore an established methodological choice.
They are not a new feature of this study. Network structure has also been
incorporated into Koopman methods. Examples include distributed geometric
learning based on graph partitions and sparsity \cite{MukherjeeEtAl2022}, as
well as nonlinear-network identification that localizes predefined dictionary
functions after identifying node neighborhoods
\cite{AnantharamanMauroy2025}.

Carleman linearization provides the corresponding polynomial viewpoint.
Polynomial systems admit infinite-dimensional linear embeddings through
monomial coordinates. Finite sections then provide tractable approximations
\cite{Carleman1932,KowalskiSteeb1991}. Convergence conditions, explicit
trajectory-error estimates for complete finite sections, and their use in
reachability analysis have been studied in general settings
\cite{AminiZhengSunMotee2025,ForetsSchilling2021}. The residual bounds developed
below do not replace that trajectory-error theory. They instead characterize
the omitted forcing of a reduced graph-induced SEIR dictionary in
epidemiological and graph quantities.

Operator methods have already addressed several epidemic questions. DMD has
been used to extract spatiotemporal patterns from infectious-disease data
\cite{ProctorEckhoff2015}, and Koopman-based prediction has been demonstrated
for COVID-19 and influenza time series \cite{MezicEtAl2024Epidemic}. Finite
observable constructions have also been used for compartmental epidemic
models. Examples include EDMD approximations of the classical SIR model
\cite{LeventidesMelasPoulios2023}, a recent SIRSD study comparing minimal and
enriched epidemiological dictionaries \cite{ZinihiEtAl2026}, Carleman
approximants for SIR dynamics \cite{MunozSanchezEtAl2025}, and a data-driven
algebraic estimate of the effective reproduction number based on Carleman
linearization \cite{MunozSanchezEtAl2026}. These studies address prediction,
approximation, inference, and dictionary design. The present work instead
examines how contact-supported nonlinearities generate an analytical reduced
dictionary whose size and closure residual can be characterized in
graph-theoretic terms.

Recent work has also connected nonlinear graph dynamics with linear dynamics on
higher-order network representations through Carleman arguments
\cite{Lacasa2026}. That preprint establishes a general structural--dynamical
connection and interprets lifted monomials as states of higher-order
combinatorial objects. The construction developed here provides an
SEIR-specific formulation for deterministic networked epidemic dynamics.

Against this background, this paper develops an analytical graph-induced
lifting for deterministic networked SEIR dynamics. The hierarchy is constructed
degree by degree. Nonlinear transmission generates the next homogeneous degree.
The linear compartmental transitions then saturate the resulting observable set
within that degree. The construction remains tied to local contact patterns and
supports the analysis of both dimensional growth and omitted next-degree
forcing. It complements earlier Koopman, Carleman, and graph-aware approaches
by deriving the reduced dictionary analytically from the contact-supported
nonlinearities and enforcing linear closedness under the linear compartmental
dynamics. The principal novelty lies in relating this linearly closed hierarchy,
its complexity, and its truncation residual explicitly to network structure.

The main contributions are as follows. First, an exact quadratic representation
separates local compartmental transitions from nonlinear transmission processes
supported by the effective channels. Second, a reduced lifting is constructed
from infection- and exposure-related channel products. Its finite-dimensional
dynamics are linear up to an explicit cubic residual. Third, a homogeneous
hierarchy is introduced in which each new degree is generated by the quadratic
transmission field and then saturated under the linear compartmental field.
Fourth, the analysis establishes dimensional-scaling results and graph-dependent
residual bounds. Finally, the numerical experiments separate support-dependent
dictionary complexity from dynamic weighting effects. They compare equal
channel intensities with equal total incoming intensity at each receiving node.

At the first lifted order, complexity is governed by the number of effective
channels. At higher fixed orders, it depends on local connectivity. When local
connectivity remains uniformly bounded, the graph-induced hierarchy grows
linearly with network size. This contrasts with complete polynomial liftings.

The framework also provides a basis for subsequent Lyapunov analysis and
predictive control. These developments are left for future work. Extensions to
other networked compartmental models are possible when nonlinear processes are
supported on identifiable channels and differentiation yields a manageable
polynomial hierarchy. Nonpolynomial incidence, explicit migration, delays, and
additional nonlinear transitions require separate constructions.

The paper is organized as follows. Section~\ref{sec:networked_SEIR_models}
introduces the networked SEIR model and its structured quadratic
representation. Section~\ref{sec:tensor_based_lifting_framework} develops the
graph-induced tensor lifting framework and the first edge-closure dynamics.
Section~\ref{sec:structural_properties_truncation_residuals} analyzes
invariance, sparsity, dimensional scaling, and residual bounds.
Section~\ref{sec:numerical_illustrations} presents numerical illustrations of
dimension, residual behavior, and higher-order dictionary growth.
Section~\ref{sec:conclusions} concludes with scope, limitations, and possible
extensions.

\section{Networked SEIR models and structured quadratic representations}
\label{sec:networked_SEIR_models}

This section introduces the mathematical formulation of the networked SEIR
model and develops an edge-based quadratic representation that explicitly
preserves the effective transmission structure. The resulting formulation
provides the foundation for the tensor-based lifting procedures presented in
Sections~\ref{sec:tensor_based_lifting_framework}--\ref{sec:numerical_illustrations}.

\subsection{Graph-theoretic formulation}

Let $\mathcal{G}=(\mathcal{V},\mathcal{E})$ be a weighted contact network
composed of $n$ interacting subpopulations. The vertex set
$\mathcal{V}=\{1,2,\ldots,n\}$ represents the collection of geographical
regions, demographic groups, or communities under consideration. The effective
transmission support $\mathcal{E}\subseteq\mathcal{V}\times\mathcal{V}$ is
written in directed form and describes the interactions through which
infections may be transmitted, both within and between subpopulations.

The strength of these interactions is characterized by a weighted adjacency
matrix $\matA=[a_{ij}]\in\mathbb{R}^{n\times n}$, where $a_{ij}\geq 0$
quantifies the epidemiological coupling from source node $j$ to receiving node
$i$. The effective transmission support is $\mathcal
E=\{(i,j)\in\mathcal V\times\mathcal V:a_{ij}>0\}$. For $i\neq j$, a pair
$(i,j)\in\mathcal{E}$ records epidemiological influence from node $j$ to node
$i$: infectious individuals in node $j$ contribute to the infection pressure
experienced by susceptible individuals in node $i$. The associated directed graph edge is $j\to i$, although the ordered pair is indexed as
$(i,j)$ to match the matrix entry $a_{ij}$ and the infection monomial $S_iI_j$.
Thus, $j$ is the source node and $i$ is the receiving or target node. Contact,
mobility, or transportation data may be used to construct the off-diagonal
coefficients, but the coefficients enter the present model only as effective
contributions to infection pressure. They do not represent explicit migration
of susceptible, exposed, infectious, or removed individuals between nodes, and
no compartment-transfer terms between subpopulations are included.

Diagonal entries are also allowed. Specifically, $a_{ii}\geq 0$ represents
effective transmission within subpopulation $i$. A diagonal pair
$(i,i)\in\mathcal{E}$ therefore does not describe self-infection at the
individual level. Instead, it aggregates contacts between susceptible and
infectious individuals belonging to the same subpopulation. 

Throughout, the between-subpopulation contact structure is written in directed
form. An undirected contact between nodes $i$ and $j$ is represented by including
both off-diagonal pairs $(i,j)$ and $(j,i)$. Their weights need not be equal,
although equal values may be used for symmetric interactions. The cardinality
$m=|\mathcal{E}|$ denotes the total number of effective transmission channels,
including diagonal and off-diagonal pairs.

\subsection{Networked SEIR dynamics}

Each node $i\in\mathcal{V}$ is associated with four state variables
$(S_i,E_i,I_i,R_i)$, representing the local susceptible, exposed, infectious,
and removed compartments, respectively.

Throughout this work, the compartmental variables are assumed to represent
normalized population fractions, so that $S_i+E_i+I_i+R_i=1$ for
$i\in\mathcal{V}$. The coefficients $a_{ij}$ therefore encode effective contact
intensities within and between subpopulations, and no additional normalization
with respect to local population sizes is required.

The networked SEIR dynamics are given by 
\begin{equation}
\begin{cases}
\dot{S}_i =
-\beta
\sum_{j=1}^{n}
a_{ij}S_iI_j,
\\
\dot{E}_i =
\beta
\sum_{j=1}^{n}
a_{ij}S_iI_j
-
\sigma E_i,
\\
\dot{I}_i =
\sigma E_i
-
\gamma I_i,
\\
\dot{R}_i =
\gamma I_i,
\end{cases}
\label{eq:seir}
\end{equation}
Here, $\beta>0$ denotes the transmission rate. The quantities $\sigma^{-1}$
and $\gamma^{-1}$ are the average latent and infectious periods, respectively.
Summing the four equations at each node gives
\begin{equation*}
\frac{\mathrm d}{\mathrm dt}\left(S_i+E_i+I_i+R_i\right)=0,
\quad i\in\mathcal V,
\end{equation*}
so every local subpopulation conserves its own compartmental mass. The coupling
modifies the force of infection but does not transfer population mass between
nodes.

Following standard epidemiological terminology, the quantity
$\lambda_i=\beta\sum_{j=1}^{n}a_{ij}I_j$ represents the force of infection
acting on node $i$, so that the incidence term may be written as
$\dot{S}_i=-\lambda_iS_i$. Its decomposition $\lambda_i=\beta
a_{ii}I_i+\beta\sum_{j\neq i}a_{ij}I_j$ separates within-subpopulation
transmission from infection pressure generated by other subpopulations.

For each compartment $X\in\{S,E,I,R\}$, define
$X=\operatorname{col}(X_1,\ldots,X_n)\in\mathbb{R}^n$. 
The complete network state vector is then 
\begin{equation}
x=\operatorname{col}(S,E,I,R)
\in
\mathbb{R}^{4n}.
\label{eq:fullstate}
\end{equation}

The resulting high-dimensional nonlinear system encodes both within- and
between-subpopulation transmission in the effective contact structure.

\subsection{Edge-based interaction structure}

The nonlinear dynamics of the networked SEIR model originate exclusively from
the infection process. In particular, the only nonlinear terms appearing in
\eqref{eq:seir} are the bilinear monomials $S_iI_j$ for $(i,j)\in\mathcal{E}$.
Off-diagonal monomials represent transmission between subpopulations, whereas
diagonal monomials $S_iI_i$ represent transmission within a subpopulation.

Let $\mathcal E=\{(i_e,j_e):e=1,\ldots,m\}$ be a fixed enumeration of the
effective transmission channels. This structure motivates the edge-based infection dictionary $\phi^{SI}(x)\in\mathbb{R}^m$, defined componentwise by
$\phi^{SI}_e(x)=S_{i_e}I_{j_e}$, $e=1,\ldots,m$. The superscript $SI$ indicates
that the dictionary contains susceptible--infectious interaction terms. Its
dimension is therefore determined by the support of the effective transmission
mechanism rather than by the dimension of the full ordered quadratic tensor
space.

\begin{proposition}[Edge-based quadratic representation]
\label{prop:edge_based_quadratic}
Let $\mathcal{G}=(\mathcal{V},\mathcal{E})$ be the weighted contact network
associated with the networked SEIR model \eqref{eq:seir}, and let
$\phi^{SI}(x)\in\mathbb{R}^m$ be the edge-based infection dictionary such that
$\phi^{SI}_e(x)=S_{i_e}I_{j_e}$ for $(i_e,j_e)\in\mathcal{E}$. Then the
networked SEIR dynamics admit the structured quadratic representation
\begin{equation}
\dot{x}=\matA_{\ell}x+\matM\phi^{SI}(x),
\label{eq:stqrep}
\end{equation}
where $\matA_{\ell}$ contains the linear compartmental transitions and
$\matM\in\mathbb{R}^{4n\times m}$ is a sparse coupling matrix. For each channel
$(i_e,j_e)\in\mathcal{E}$, the $e$-th column of $\matM$ has only two nonzero
entries. Specifically, $\matM_{i_e,e}=-\beta a_{i_ej_e}$, whereas
$\matM_{n+i_e,e}=+\beta a_{i_ej_e}$. All other entries of that column are zero.
Consequently, the nonlinear part of the model is completely determined by
$m=|\mathcal{E}|$ transmission-supported bilinear monomials.
\end{proposition}

\begin{proof}
The only nonlinear terms in \eqref{eq:seir} arise from the infection process.
For each effective transmission channel $(i_e,j_e)\in\mathcal{E}$, including
the diagonal case $i_e=j_e$, the corresponding bilinear interaction is
$S_{i_e}I_{j_e}$. This interaction decreases the susceptible compartment of
node $i_e$ at rate $\beta a_{i_ej_e}S_{i_e}I_{j_e}$ and increases the exposed
compartment of the same node by the same amount. Therefore, the $e$-th
interaction contributes $-\beta a_{i_ej_e}\phi^{SI}_e(x)$ to the equation for
$S_{i_e}$ and $+\beta a_{i_ej_e}\phi^{SI}_e(x)$ to the equation for $E_{i_e}$,
with no direct contribution to any other compartment. Summing over all
effective transmission channels yields the structured quadratic representation
\eqref{eq:stqrep}. Since all remaining terms in the SEIR dynamics are linear
transitions between compartments, the dictionary $\phi^{SI}(x)$ completely
characterizes the nonlinear part of the model.  \end{proof}

Unlike conventional polynomial lifting procedures based on the full ordered
quadratic tensor, the proposed dictionary retains only the interactions
supported by the effective transmission structure. The locality of epidemic transmission induces an intrinsic sparsity pattern
that can be exploited analytically and computationally.

\subsection{Structured quadratic representation}

The edge-based interaction dictionary allows the networked SEIR model to be
represented as a structured quadratic dynamical system.

The linear part of the dynamics is given by
\begin{equation*}
\matA_{\ell}
=
\begin{bmatrix}
\matzero & \matzero & \matzero & \matzero\\
\matzero & -\sigma \Id_n & \matzero & \matzero\\
\matzero & \sigma \Id_n & -\gamma \Id_n & \matzero\\
\matzero & \matzero & \gamma \Id_n & \matzero
\end{bmatrix}
\in
\mathbb{R}^{4n\times 4n},
\end{equation*}
where $\Id_n$ denotes the identity matrix of dimension $n$.

The nonlinear infection terms are collected through the sparse matrix
$\matM\in\mathbb{R}^{4n\times m}$ and the infection dictionary $\phi^{SI}(x)$
defined in Proposition~\ref{prop:edge_based_quadratic}. 
The complete system may therefore be written as \eqref{eq:stqrep}.

An equivalent representation may be obtained by introducing a matrix
$\matH\in\mathbb{R}^{4n\times (4n)^2}$ such that 
\begin{equation*}
\dot{x}
=
\matA_{\ell}x
+
\matH(x\otimes x),
\end{equation*}
where $\otimes$ denotes the Kronecker product~\cite{HornJohnson2012}.

However, the explicit construction of $\matH$ generally embeds the dynamics
into the full ordered quadratic tensor space, whose dimension grows
quadratically with the state dimension. By contrast, the structured
representation \eqref{eq:stqrep} retains only the effective infection processes
supported by the transmission network and therefore preserves the intrinsic
locality of epidemic transmission.

\subsection{A three-node mobility-corridor example}
\label{sec:mobility-corridor-example}

The construction can be illustrated by a simple regional mobility example.
Consider three urban areas connected by a mobility corridor.
Nodes $1$ and $3$ represent smaller peripheral cities, while node $2$
represents an intermediate city that concentrates employment, transport
connections, or shared services. Contacts among individuals within each city
generate within-subpopulation transmission, while daily commuting and
short-distance mobility generate between-subpopulation transmission between
cities $1$ and $2$, and between cities $2$ and $3$. Direct mobility between
cities $1$ and $3$ is assumed to be negligible at the spatial resolution of the
model.

Following the convention introduced above, assume one positive diagonal channel
for each city and two directed off-diagonal channels for each reciprocal
mobility connection. The effective transmission support is 
\(
\mathcal E=\{(1,1),(1,2),(2,1),(2,2),(2,3),(3,2),(3,3)\}.
\)

For the construction below, these channels are enumerated in the displayed
order, grouped by receiving city.  For $i\neq j$, a pair $(i,j)$ means that
infectious individuals associated with city $j$ contribute to the infection
pressure experienced by susceptible individuals in city $i$. For $i=j$, the
pair represents aggregate transmission among individuals belonging to the same
city. The off-diagonal weights need not be symmetric, since commuting flows,
contact intensities, or exposure opportunities may differ in the two
directions.

The susceptible and exposed equations are then 
\[
\begin{aligned}
\dot S_1 &= -\beta S_1(a_{11}I_1+a_{12}I_2), &
\dot E_1 &= \beta S_1(a_{11}I_1+a_{12}I_2)-\sigma E_1,\\
\dot S_2 &= -\beta S_2(a_{21}I_1+a_{22}I_2+a_{23}I_3), &
\dot E_2 &= \beta S_2(a_{21}I_1+a_{22}I_2+a_{23}I_3)-\sigma E_2,\\
\dot S_3 &= -\beta S_3(a_{32}I_2+a_{33}I_3), &
\dot E_3 &= \beta S_3(a_{32}I_2+a_{33}I_3)-\sigma E_3,
\end{aligned}
\]
while $\dot I_i=\sigma E_i-\gamma I_i$ and $\dot R_i=\gamma I_i$ for $i=1,2,3$.

The complete state is ordered as
\(
x=\operatorname{col}(S_1,S_2,S_3,E_1,E_2,E_3,I_1,I_2,I_3,R_1,R_2,R_3).
\)
The nonlinear infection monomials are therefore $S_1I_1$, $S_1I_2$, $S_2I_1$,
$S_2I_2$, $S_2I_3$, $S_3I_2$, and $S_3I_3$. The monomials $S_1I_3$ and $S_3I_1$
do not appear because cities $1$ and $3$ do not have a direct effective contact
in this model. With the edge ordering above, the edge-based infection
dictionary is 
\(
\phi^{SI}(x)=\operatorname{col}(S_1I_1,S_1I_2,S_2I_1,S_2I_2,S_2I_3,S_3I_2,S_3I_3).
\)

Collect the nonzero incoming weights of each receiving city, in the ordering
used in \(\phi^{SI}(x)\), into the vectors
\({a}_1=\operatorname{col}(a_{11},a_{12})\in\mathbb{R}^2\),
\({a}_2=\operatorname{col}(a_{21},a_{22},a_{23})\in\mathbb{R}^3\), and
\({a}_3=\operatorname{col}(a_{32},a_{33})\in\mathbb{R}^2\).
For this three-node network,
\[
\matA_{\ell}=
\begin{bmatrix}
\matzero & \matzero & \matzero & \matzero \\
\matzero & -\sigma \Id_3 & \matzero & \matzero \\
\matzero & \sigma \Id_3 & -\gamma \Id_3 & \matzero \\
\matzero & \matzero & \gamma \Id_3 & \matzero
\end{bmatrix},
\quad
\matM=
\begin{bmatrix}
-\matB \\
\matB \\
\matzero \\
\matzero
\end{bmatrix},
\]
where
\(
\matB
=
\beta
\left(
{a}_1^{\top}
\oplus
{a}_2^{\top}
\oplus
{a}_3^{\top}
\right)
\in\mathbb{R}^{3\times7},
\)
and \(\oplus\) denotes the block-diagonal matrix direct sum
\cite{HornJohnson2012}.

The columns of $\matB$ correspond, in order, to the monomials $S_1I_1$,
$S_1I_2$, $S_2I_1$, $S_2I_2$, $S_2I_3$, $S_3I_2$, and $S_3I_3$. Thus, the
susceptible rows of $\matM$ carry the negative nonlinear contributions, the
exposed rows carry the positive ones, and the infectious and removed rows
receive no direct nonlinear term. Each column contributes only to the
susceptible and exposed equations of its receiving city. Therefore, $\dot
x=\matA_{\ell}x+\matM\phi^{SI}(x)$ gives the complete structured quadratic
representation \eqref{eq:stqrep} for the three-node mobility-corridor network.
Bilinear products such as $S_1I_3$ and $S_3I_1$ are algebraically possible,
but they are not included because they are not supported by the effective
transmission matrix. 

The example shows how the edge-based quadratic representation accommodates
both within- and between-subpopulation infection processes while retaining only
transmission-supported interactions. It also
anticipates the lifting construction of
Section~\ref{sec:tensor_based_lifting_framework}: differentiating the infection
observables $S_iI_j$ generates exposure-type observables $S_iE_j$ on the same
effective transmission channels, which motivates the first edge-closure
dictionary.

\section{Tensor-based lifting framework}
\label{sec:tensor_based_lifting_framework}

This section develops the graph-induced tensor lifting framework. Starting from
the structured quadratic representation in
Section~\ref{sec:networked_SEIR_models}, the SEIR vector field is separated
into its linear compartmental part and its quadratic transmission part. This
separation distinguishes two operations that must be treated differently in the
lifted hierarchy: the linear dynamics propagate observables within a fixed
polynomial degree, whereas the quadratic transmission dynamics generate
observables of the next degree. The construction closes every homogeneous dictionary under the linear SEIR transitions before using the
quadratic mechanism to generate the next block.

\subsection{Tensor notation and polynomial-degree structure}

Let $x\in\mathbb{R}^{4n}$ denote the full networked SEIR state vector in
\eqref{eq:fullstate}. Denote by $x^{\otimes k}$ the $k$-fold Kronecker product,
defined recursively by
\begin{equation}
x^{\otimes k}
=
x^{\otimes(k-1)}\otimes x,
\quad k\geq 2.
\label{eq:ordered_kronecker_tensor}
\end{equation}
The vector $x^{\otimes k}$ is the complete ordered Kronecker tensor of degree
$k$ and has coordinate dimension $(4n)^k$. Because scalar multiplication is
commutative, different ordered coordinates may represent the same polynomial
monomial. A complete symmetric degree-$k$ monomial basis removes these
permutation duplicates and has dimension
\begin{equation}
\binom{4n+k-1}{k}.
\label{eq:symmetric_degree_k_dimension}
\end{equation}
The ordered representation in \eqref{eq:ordered_kronecker_tensor} is retained
because it provides a convenient ambient coordinate space for sparse selection
operators, while the graph-induced dictionaries themselves contain distinct
commutative monomials.

Using \eqref{eq:stqrep}, decompose the SEIR vector field as
\begin{equation}
f(x)=f_{\mathrm L}(x)+f_{\mathrm Q}(x),
\quad
f_{\mathrm L}(x)=\matA_{\ell}x,
\quad
f_{\mathrm Q}(x)=\matM\phi^{SI}(x).
\label{eq:linear_quadratic_vector_field_split}
\end{equation}
Here $f_{\mathrm L}$ contains the linear compartmental transitions, whereas
$f_{\mathrm Q}$ contains the quadratic transmission terms supported by the
effective transmission channels.

For any polynomial observable $q$, its Lie derivative along a vector field $g$
is
\begin{equation}
\mathcal{L}_g q(x)
=
\nabla q(x)^{\top}g(x),
\label{eq:lie_derivative_definition}
\end{equation}
see e.g. \cite{Isidori1995}.
Along trajectories, $\frac{d}{dt}q(x(t))=\mathcal{L}_f q(x(t))$, and
linearity of the Lie derivative gives
$\mathcal{L}_f q=\mathcal{L}_{f_{\mathrm L}}q+
\mathcal{L}_{f_{\mathrm Q}}q$.

\begin{lemma}[Degree propagation under the linear and quadratic fields]
\label{lem:degree_propagation}
If $q$ is a homogeneous polynomial of degree $k\geq 1$, then
$\mathcal{L}_{f_{\mathrm L}}q$ is homogeneous of degree $k$, whereas
$\mathcal{L}_{f_{\mathrm Q}}q$ is homogeneous of degree $k+1$.
\end{lemma}

\begin{proof}
The gradient of a homogeneous degree-$k$ polynomial is homogeneous of degree
$k-1$. Since $f_{\mathrm L}$ is linear and $f_{\mathrm Q}$ is quadratic, the
products $\nabla q^{\top}f_{\mathrm L}$ and
$\nabla q^{\top}f_{\mathrm Q}$ have degrees $k$ and $k+1$, respectively.
\end{proof}

Lemma~\ref{lem:degree_propagation} is the organizing principle of the
hierarchy. Same-degree terms generated by the linear compartmental dynamics
belong to the retained homogeneous block and only the quadratic transmission
part creates the next polynomial degree.

\subsection{Compartment selection vectors}

For each node $i\in\mathcal{V}$, define the compartment selection vectors
$s_i,\eta_i,\iota_i,\rho_i\in\mathbb{R}^{4n}$ by
\begin{equation*}
S_i=s_i^{\top}x,
\quad
E_i=\eta_i^{\top}x,
\quad
I_i=\iota_i^{\top}x,
\quad
R_i=\rho_i^{\top}x.
\end{equation*}
For an effective transmission channel $(i,j)\in\mathcal{E}$,
\begin{equation*}
S_iI_j
=
(s_i\otimes\iota_j)^{\top}(x\otimes x),
\quad
S_iE_j
=
(s_i\otimes\eta_j)^{\top}(x\otimes x).
\end{equation*}

\subsection{Linear saturation of homogeneous monomial sets}

For a polynomial $p$, let $\operatorname{Mon}_k(p)$ denote the set of distinct
commutative degree-$k$ monomials that occur in $p$ with nonzero coefficient
after algebraically identical terms have been combined. Factor permutations
are identified, repeated factors are retained, and coefficient multiplicity
does not affect membership in the set.

For a finite set $\mathcal{A}$ of homogeneous monomials and a vector field
$f$, define the one-step degree-$k$ generation operator by
\begin{equation*}
\Gamma_f^{(k)}(\mathcal{A})
:=
\bigcup_{q\in\mathcal{A}}
\operatorname{Mon}_k\!\left(\mathcal{L}_f q\right).
\end{equation*}

For a finite set $\mathcal{A}$ of homogeneous degree-$k$ monomials and a
vector field $f$, define its degree-$k$ saturation under $f$ by
\begin{equation}
\operatorname{Sat}_{f}^{(k)}(\mathcal{A})
:=
\bigcup_{r\geq 0}\mathcal{A}_r,
\label{eq:saturation_operator}
\end{equation}
where the sequence $(\mathcal{A}_r)_{r\geq 0}$ is generated recursively by
\begin{equation}
\mathcal{A}_0=\mathcal{A},
\quad
\mathcal{A}_{r+1}
=
\mathcal{A}_r
\cup
\Gamma_{f}^{(k)}(\mathcal{A}_r),
\quad r\geq 0.
\label{eq:saturation_iteration}
\end{equation}
Equivalently, $\operatorname{Sat}_{f}^{(k)}(\mathcal{A})$ is the smallest
set of degree-$k$ commutative monomials containing $\mathcal{A}$ such that
$\operatorname{span}\operatorname{Sat}_{f}^{(k)}(\mathcal{A})$ is invariant
under $\mathcal{L}_{f}$. Thus, saturation produces a monomial set, and its linear span is the
corresponding invariant subspace.

\begin{lemma}[Finiteness of the SEIR linear saturation]
\label{lem:finite_linear_closure}
For every finite set $\mathcal{A}$ of homogeneous degree-$k$ monomials in the
variables $S$, $E$, and $I$, the saturation
$\operatorname{Sat}_{f_{\mathrm L}}^{(k)}(\mathcal{A})$ is finite. Moreover,
\begin{equation}
\left|\operatorname{Sat}_{f_{\mathrm L}}^{(k)}(\mathcal{A})\right|
\leq
2^k|\mathcal{A}|.
\label{eq:linear_closure_cardinality_bound}
\end{equation}
Saturation preserves the node indices and susceptible factors of each seed
monomial and introduces no removed-compartment factors.
\end{lemma}

\begin{proof}
Under $f_{\mathrm L}$, susceptible variables have zero derivative, exposed
variables reproduce themselves through $\dot E_i=-\sigma E_i$, and an
infectious variable produces a linear combination of $I_i$ and $E_i$ through
$\dot I_i=\sigma E_i-\gamma I_i$. Therefore, applying
$\mathcal{L}_{f_{\mathrm L}}$ to a monomial can only replace selected
infectious factors $I_i$ by exposed factors $E_i$ with the same node index, and
it cannot introduce new node indices, susceptible factors, or removed factors.
If a degree-$k$ seed contains $r\leq k$ infectious factor occurrences, at most
$2^r\leq 2^k$ replacement patterns are possible. Commutative
canonicalization and repeated indices can only reduce this number. Taking the
union over the seeds gives \eqref{eq:linear_closure_cardinality_bound}, and
the sequence in \eqref{eq:saturation_iteration} stabilizes after finitely
many steps.
\end{proof}

\subsection{Graph-induced quadratic dictionary}

Fix an enumeration
$\mathcal{E}=\{(i_e,j_e):e=1,\ldots,m\}$ of the effective transmission
channels, including diagonal channels when present. The transmission-supported
quadratic seed set is 
\begin{equation}
\mathcal{G}_2^{\mathcal{E}}
:=
\left\{S_iI_j:(i,j)\in\mathcal{E}\right\}.
\label{eq:quadratic_seed_set}
\end{equation}
Its linear saturation is the degree-2 graph-induced dictionary
\begin{equation}
\mathcal{M}_2^{\mathcal{E}}
:=
\operatorname{Sat}_{f_{\mathrm L}}^{(2)}
\!\left(\mathcal{G}_2^{\mathcal{E}}\right)
=
\left\{S_iI_j,\ S_iE_j:(i,j)\in\mathcal{E}\right\}.
\label{eq:graph_induced_monomial_base}
\end{equation}
Indeed,
$\mathcal{L}_{f_{\mathrm L}}(S_iI_j)
=-\gamma S_iI_j+\sigma S_iE_j$ and
$\mathcal{L}_{f_{\mathrm L}}(S_iE_j)=-\sigma S_iE_j$.
Thus, the exposure observables are not an ad hoc enlargement of the infection
dictionary. They are exactly the coordinates required to close the quadratic
block under the linear SEIR transitions.

Define the infection selection operator
$T_{SI}^{\mathcal{E}}\in\mathbb{R}^{m\times(4n)^2}$ by
\begin{equation*}
\left(T_{SI}^{\mathcal{E}}\right)_{e,:}
=
(s_{i_e}\otimes\iota_{j_e})^{\top},
\quad e=1,\ldots,m,
\end{equation*}
so that
\begin{equation*}
\phi^{SI}(x)
=
T_{SI}^{\mathcal{E}}(x\otimes x)
=
\operatorname{col}
\left(S_{i_1}I_{j_1},\ldots,S_{i_m}I_{j_m}\right).
\end{equation*}
Similarly, define
$T_{SE}^{\mathcal{E}}\in\mathbb{R}^{m\times(4n)^2}$ by
\begin{equation*}
\left(T_{SE}^{\mathcal{E}}\right)_{e,:}
=
(s_{i_e}\otimes\eta_{j_e})^{\top},
\quad e=1,\ldots,m,
\end{equation*}
\begin{equation*}
\phi^{SE}(x)
=
T_{SE}^{\mathcal{E}}(x\otimes x)
=
\operatorname{col}
\left(S_{i_1}E_{j_1},\ldots,S_{i_m}E_{j_m}\right).
\end{equation*}
The graph-induced quadratic selection operator and observable vector are
\begin{equation}
T_2^{\mathcal{E}}
=
\begin{bmatrix}
T_{SI}^{\mathcal{E}}\\
T_{SE}^{\mathcal{E}}
\end{bmatrix},
\quad
\Phi_2^{\mathcal{E}}(x)
=
T_2^{\mathcal{E}}(x\otimes x)
=
\begin{bmatrix}
\phi^{SI}(x)\\
\phi^{SE}(x)
\end{bmatrix}
\in\mathbb{R}^{2m}.
\label{eq:phi2_edge_dictionary}
\end{equation}
Each transmission channel contributes two monomials, one of type $SI$ and one
of type $SE$. Hence $|\mathcal{M}_2^{\mathcal{E}}|=2m$.

\subsection{First edge-closure lifting and residual dynamics}

The first edge-closure lifted state is
\begin{equation}
z_2^{\mathcal{E}}(x)
=
\begin{bmatrix}
x\\
\Phi_2^{\mathcal{E}}(x)
\end{bmatrix}
\in
\mathbb{R}^{4n+2m}.
\label{eq:first_edge_lift_tensor}
\end{equation}
For every effective transmission channel $(i,j)\in\mathcal{E}$,
\begin{equation}
\mathcal{L}_f(S_iI_j)
=
-\gamma S_iI_j
+
\sigma S_iE_j
-
\beta
\sum_{\ell:(i,\ell)\in\mathcal{E}}
a_{i\ell}S_iI_{\ell}I_j,
\label{eq:edge_SI_derivative_tensor}
\end{equation}
and
\begin{equation}
\mathcal{L}_f(S_iE_j)
=
-\sigma S_iE_j
-
\beta
\sum_{\ell:(i,\ell)\in\mathcal{E}}
a_{i\ell}S_iI_{\ell}E_j
+
\beta
\sum_{\ell:(j,\ell)\in\mathcal{E}}
a_{j\ell}S_iS_jI_{\ell}.
\label{eq:edge_SE_derivative_tensor}
\end{equation}
The degree-2 terms in these equations are represented exactly because
$\mathcal{M}_2^{\mathcal{E}}$ is linearly closed. The remaining terms are
homogeneous cubic contributions generated by $f_{\mathrm Q}$. Diagonal
channels require no separate formula and may produce repeated-factor monomials
such as $S_iI_i^2$, $S_iI_iE_i$, and $S_i^2I_i$.

\begin{proposition}[First edge-closure lifting]
\label{prop:first_edge_closure_lifting}
The lifted state in \eqref{eq:first_edge_lift_tensor} satisfies
\begin{equation}
\frac{d}{dt}z_2^{\mathcal{E}}(x)
=
\matK_2^{\mathcal{E}}z_2^{\mathcal{E}}(x)
+
r_2^{\mathcal{E}}(x),
\label{eq:lifted_dynamics_tensor}
\end{equation}
for a graph-structured matrix
$\matK_2^{\mathcal{E}}\in
\mathbb{R}^{(4n+2m)\times(4n+2m)}$, where
\begin{equation}
r_2^{\mathcal{E}}(x)
=
\begin{bmatrix}
0\\
r_{SI}^{\mathcal{E}}(x)\\
r_{SE}^{\mathcal{E}}(x)
\end{bmatrix},
\label{eq:residual_block_structure}
\end{equation}
and, for $e=1,\ldots,m$,
\begin{equation}
\left(r_{SI}^{\mathcal{E}}(x)\right)_e
=
-\beta
\sum_{\ell:(i_e,\ell)\in\mathcal{E}}
a_{i_e\ell}S_{i_e}I_{\ell}I_{j_e},
\label{eq:first_residual_si}
\end{equation}
\begin{equation}
\left(r_{SE}^{\mathcal{E}}(x)\right)_e
=
-\beta
\sum_{\ell:(i_e,\ell)\in\mathcal{E}}
a_{i_e\ell}S_{i_e}I_{\ell}E_{j_e}
+
\beta
\sum_{\ell:(j_e,\ell)\in\mathcal{E}}
a_{j_e\ell}S_{i_e}S_{j_e}I_{\ell}.
\label{eq:first_residual_se}
\end{equation}
Every nonzero component of $r_2^{\mathcal{E}}$ is homogeneous of degree three.
\end{proposition}

\begin{proof}
The $x$-block follows from Proposition~\ref{prop:edge_based_quadratic}. The
linear parts of \eqref{eq:edge_SI_derivative_tensor} and
\eqref{eq:edge_SE_derivative_tensor} remain in the linearly closed quadratic
block. Their quadratic-transmission contributions are exactly the cubic terms
in \eqref{eq:first_residual_si} and \eqref{eq:first_residual_se}. Collecting
the retained coefficients gives $\matK_2^{\mathcal{E}}$, while the omitted
cubic terms give \eqref{eq:residual_block_structure}.
\end{proof}

\subsection{Linearly closed higher-order dictionaries}

Assume that the degree-$k$ dictionary
$\mathcal{M}_k^{\mathcal{E}}$ has been constructed and is linearly closed.
The quadratic transmission field generates the degree-$(k+1)$ seed set
\begin{equation}
\mathcal{G}_{k+1}^{\mathcal{E}}
:=
\Gamma_{f_{\mathrm Q}}^{(k+1)}
\!\left(\mathcal{M}_k^{\mathcal{E}}\right),
\quad k\geq 2.
\label{eq:next_degree_seed_generation}
\end{equation}
The next graph-induced dictionary is its linear saturation,
\begin{equation}
\mathcal{M}_{k+1}^{\mathcal{E}}
:=
\operatorname{Sat}_{f_{\mathrm L}}^{(k+1)}
\!\left(\mathcal{G}_{k+1}^{\mathcal{E}}\right),
\quad k\geq 2.
\label{eq:graph_induced_monomial_recursion}
\end{equation}
Thus, the quadratic field creates the next degree and the linear field closes
that new homogeneous block. By Lemma~\ref{lem:finite_linear_closure}, every
set in the recursion is finite.

At degree three, the seed set produced from
$\mathcal{M}_2^{\mathcal{E}}$ contains the transmission-supported monomials
$S_iI_{\ell}I_j$, $S_iI_{\ell}E_j$, and $S_iS_jI_{\ell}$ with the channel
conditions in 
\eqref{eq:edge_SI_derivative_tensor}--\eqref{eq:edge_SE_derivative_tensor}. 
Its linear saturation additionally contains all same-degree variants obtained by
replacing infectious factors by exposed factors at the same node indices. This saturation is essential: it ensures that differentiating the retained cubic
observables cannot generate an omitted cubic term through the linear
compartmental transitions.

With any fixed ordering of $\mathcal{M}_k^{\mathcal{E}}$, let
$T_k^{\mathcal{E}}$ be a sparse selection operator that extracts one ordered
tensor coordinate as a representative of each commutative monomial. Define
\begin{equation*}
\Phi_k^{\mathcal{E}}(x)
=
\bigl[q(x)\bigr]_{q\in\mathcal{M}_k^{\mathcal{E}}}
=
T_k^{\mathcal{E}}x^{\otimes k},
\quad k\geq 2.
\end{equation*}
Then
$\dim\Phi_k^{\mathcal{E}}=|\mathcal{M}_k^{\mathcal{E}}|$ and different ordered
representatives of the same commutative product are not counted separately.
When derivative equations are assembled, coefficients of identical canonical
monomials are combined within each scalar equation.

\begin{proposition}[Exact degree-coupled graph-induced hierarchy]
\label{prop:exact_degree_coupled_hierarchy}
For every $k\geq 2$, there exist constant matrices
$\matK_{k,k}^{\mathcal{E}}$ and $\matK_{k,k+1}^{\mathcal{E}}$ such that
\begin{equation}
\frac{d}{dt}\Phi_k^{\mathcal{E}}(x)
=
\matK_{k,k}^{\mathcal{E}}\Phi_k^{\mathcal{E}}(x)
+
\matK_{k,k+1}^{\mathcal{E}}\Phi_{k+1}^{\mathcal{E}}(x).
\label{eq:exact_homogeneous_block_dynamics}
\end{equation}
The first term represents the linear SEIR transitions within degree $k$, and
the second represents the quadratic-transmission generation of degree $k+1$.
\end{proposition}

\begin{proof}
Linear saturation gives
$\mathcal{L}_{f_{\mathrm L}}
\operatorname{span}\mathcal{M}_k^{\mathcal{E}}
\subseteq
\operatorname{span}\mathcal{M}_k^{\mathcal{E}}$.
The seed definition \eqref{eq:next_degree_seed_generation}, followed by linear
saturation, ensures that every monomial in
$\mathcal{L}_{f_{\mathrm Q}}q$ for
$q\in\mathcal{M}_k^{\mathcal{E}}$ belongs to
$\mathcal{M}_{k+1}^{\mathcal{E}}$. Since the vector-field coefficients are
constant, collecting the corresponding coefficients yields the two matrices in
\eqref{eq:exact_homogeneous_block_dynamics}.
\end{proof}

For a prescribed order $d\geq 2$, define
\begin{equation}
z_d^{\mathcal{E}}(x)
=
\begin{bmatrix}
x\\
\Phi_2^{\mathcal{E}}(x)\\
\Phi_3^{\mathcal{E}}(x)\\
\vdots\\
\Phi_d^{\mathcal{E}}(x)
\end{bmatrix}
=
\begin{bmatrix}
x\\
T_2^{\mathcal{E}}x^{\otimes2}\\
T_3^{\mathcal{E}}x^{\otimes3}\\
\vdots\\
T_d^{\mathcal{E}}x^{\otimes d}
\end{bmatrix}.
\label{eq:higher_order_tensor_lift}
\end{equation}
The original-state block satisfies
$\dot x=\matA_{\ell}x+\matK_{1,2}^{\mathcal{E}}
\Phi_2^{\mathcal{E}}(x)$ for a suitable sparse matrix
$\matK_{1,2}^{\mathcal{E}}$. Combining this identity with
Proposition~\ref{prop:exact_degree_coupled_hierarchy} gives an exact
block-upper-bidiagonal infinite hierarchy by polynomial degree.

\begin{proposition}[Order-$d$ truncation with homogeneous residual]
\label{prop:order_d_homogeneous_residual}
For every $d\geq 2$, the finite lifted state in
\eqref{eq:higher_order_tensor_lift} satisfies
\begin{equation}
\frac{d}{dt}z_d^{\mathcal{E}}(x)
=
\matK_d^{\mathcal{E}}z_d^{\mathcal{E}}(x)
+
r_d^{\mathcal{E}}(x),
\label{eq:higher_order_tensor_dynamics}
\end{equation}
where $\matK_d^{\mathcal{E}}$ contains the retained within-degree and
next-degree couplings through degree $d$, and
\begin{equation}
r_d^{\mathcal{E}}(x)
=
\operatorname{col}
\left(
0,\ldots,0,
\mathcal{L}_{f_{\mathrm Q}}
\Phi_d^{\mathcal{E}}(x)
\right).
\label{eq:order_d_residual_last_block}
\end{equation}
Every nonzero residual component is homogeneous of degree $d+1$.
\end{proposition}

\begin{proof}
For degrees below $d$, the quadratic-transmission terms are represented in the
next retained block by Proposition~\ref{prop:exact_degree_coupled_hierarchy}.
The degree-$d$ linear terms remain in the linearly closed block
$\Phi_d^{\mathcal{E}}$, whereas its quadratic-transmission terms have degree
$d+1$ by Lemma~\ref{lem:degree_propagation} and are not retained in the
order-$d$ state. Placing these terms in the final residual block yields
\eqref{eq:order_d_residual_last_block}.
\end{proof}

For $d=2$, Proposition~\ref{prop:order_d_homogeneous_residual} reduces to
Proposition~\ref{prop:first_edge_closure_lifting}, and
\eqref{eq:order_d_residual_last_block} reproduces the cubic residual in
\eqref{eq:first_residual_si}--\eqref{eq:first_residual_se}.

\subsection{Structured Carleman and Koopman interpretation}

The construction is a graph-structured Carleman hierarchy with restricted,
dynamics-generated homogeneous blocks. Unlike complete Carleman embeddings,
it retains only graph-supported seeds and the same-degree coordinates required
by the linear compartmental dynamics. In ordered tensor notation, the complete
benchmark uses all $x^{\otimes k}$, whereas the proposed lifting uses
$T_k^{\mathcal{E}}x^{\otimes k}$. It may equivalently be viewed as an
analytically generated Koopman dictionary whose retained evolution is linear
and whose truncation residual is the omitted degree-$(d+1)$ forcing.

\subsection{Recursive construction of the lifted model}

Algorithm~\ref{alg:recursive_lifted_model_construction} summarizes the construction.

\refstepcounter{algorithm}\label{alg:recursive_lifted_model_construction}%
\setlength\LTleft{0pt}
\setlength\LTright{0pt}
\begin{longtable}{@{}r@{\hspace{0.02\textwidth}}
                    >{\raggedright\arraybackslash}p{0.91\textwidth}@{}}
\toprule
\multicolumn{2}{@{}p{0.96\textwidth}@{}}{\textbf{Algorithm \thealgorithm} (Linearly closed
construction of the graph-induced lifting up to order $d$)}\\
\midrule
\endfirsthead
\toprule
\multicolumn{2}{@{}p{0.96\textwidth}@{}}{\textbf{Algorithm \thealgorithm} (continued)}\\
\midrule
\endhead
\midrule
\multicolumn{2}{r@{}}{\textit{Continued on the next page}}\\
\endfoot
\bottomrule
\endlastfoot
1 & Decompose the SEIR vector field as
    $f=f_{\mathrm L}+f_{\mathrm Q}$ using
    \eqref{eq:linear_quadratic_vector_field_split}.\\

2 & Enumerate the effective transmission channels and form the quadratic seed
    set $\mathcal{G}_2^{\mathcal{E}}$ in
    \eqref{eq:quadratic_seed_set}. Apply $\operatorname{Sat}_{f_{\mathrm L}}^{(2)}$ to obtain
    $\mathcal{M}_2^{\mathcal{E}}$ in
    \eqref{eq:graph_induced_monomial_base}.\\

3 & For each $k=2,\ldots,d-1$, apply
    $\Gamma_{f_{\mathrm Q}}^{(k+1)}$ to
    $\mathcal{M}_k^{\mathcal{E}}$ and form the seed set
    $\mathcal{G}_{k+1}^{\mathcal{E}}$ using
    \eqref{eq:next_degree_seed_generation}.\\

4 & Apply $\operatorname{Sat}_{f_{\mathrm L}}^{(k+1)}$ to
    $\mathcal{G}_{k+1}^{\mathcal{E}}$ using
    \eqref{eq:saturation_operator}--
    \eqref{eq:saturation_iteration}. The resulting saturated set is
    $\mathcal{M}_{k+1}^{\mathcal{E}}$ in
    \eqref{eq:graph_induced_monomial_recursion}.\\

5 & At every generation and saturation step, canonicalize products as commutative
    monomials, retain repeated factors, merge duplicate dictionary entries, and
    combine coefficients of identical canonical terms within each scalar
    derivative equation.\\

6 & Choose a deterministic ordering of each
    $\mathcal{M}_k^{\mathcal{E}}$, construct the sparse selection operator
    $T_k^{\mathcal{E}}$, and form
    $\Phi_k^{\mathcal{E}}(x)=T_k^{\mathcal{E}}x^{\otimes k}$.\\

7 & Assemble $z_d^{\mathcal{E}}(x)$ and the block-upper-bidiagonal retained
    operator $\matK_d^{\mathcal{E}}$. Place only
    $\mathcal{L}_{f_{\mathrm Q}}\Phi_d^{\mathcal{E}}(x)$ in the final
    residual block, as in \eqref{eq:order_d_residual_last_block}.\\

8 & Verify that every retained degree is linearly closed by checking that
    $\operatorname{Mon}_k(\mathcal{L}_{f_{\mathrm L}}q)
    \subseteq\mathcal{M}_k^{\mathcal{E}}$ for all
    $q\in\mathcal{M}_k^{\mathcal{E}}$.\\
\end{longtable}

\begin{remark}[Three-subpopulation illustration]
Algorithm~\ref{alg:recursive_lifted_model_construction} separates
next-degree generation by quadratic transmission from same-degree
saturation under the linear compartmental dynamics. Returning to the
three-subpopulation mobility-corridor example of
Subsection~\ref{sec:mobility-corridor-example}, the seven effective
transmission channels generate 7 quadratic seeds, whose linear saturation
contains 14 quadratic monomials. At the next degree, they produce 42
distinct cubic seeds and 67 saturated cubic monomials. Consequently,
\(\dim z_2^{\mathcal E}=26\), \(\dim z_3^{\mathcal E}=93\),
and the order-three residual contains only degree-four terms.
For comparison, the state dimension is \(12\), so the corresponding
complete ordered Kronecker liftings through degrees two and three would
have dimensions \(156\) and \(1884\), respectively. 
The graph-induced construction therefore retains only the
coordinates generated by the effective transmission channels and required
by their subsequent linear evolution, rather than all ordered polynomial
products.
\end{remark}

This example motivates the analysis in
Section~\ref{sec:structural_properties_truncation_residuals}, which quantifies
the reduction in dictionary growth relative to complete polynomial liftings.

\section{Structural properties and truncation residuals}
\label{sec:structural_properties_truncation_residuals}

This section analyzes the graph-induced lifted representations constructed in
Section~\ref{sec:tensor_based_lifting_framework}. The admissible epidemiological
state space and its positive invariance are considered first. The dimensional
analysis is then reformulated consistently with the two-stage dictionary
construction: the quadratic transmission field generates the next homogeneous
degree, while the linear SEIR field closes that degree without introducing new
node indices. This separation yields fixed-order linear scaling on uniformly
bounded-in-degree network families, with order-dependent constants that account
explicitly for linear saturation. The first edge-closure residual is bounded in
terms of weighted incoming transmission intensity, and the higher-order
residual is shown to be a bounded homogeneous forcing of degree one above the
truncation order. Diagonal within-subpopulation channels are included
throughout whenever they belong to the effective transmission support.

\subsection{Admissible epidemiological state space}

Since the compartmental variables are normalized population fractions, the
natural state space of the networked SEIR model is
\begin{equation}
\Omega
=
\left\{
x\in\mathbb{R}_{\geq 0}^{4n}:
S_i+E_i+I_i+R_i=1,\quad i=1,\ldots,n
\right\}.
\label{eq:admissible_state_space}
\end{equation}
The set $\Omega$ is the product of $n$ local probability simplices and
represents the epidemiologically admissible region. Throughout this section, the parameters satisfy the conditions introduced with
the model. In particular, $\beta>0$,
$\sigma>0$, $\gamma>0$, and $a_{ij}\geq 0$.

\begin{proposition}[Positive invariance of the admissible state space]
\label{prop:positive_invariance}
The set $\Omega$ defined in \eqref{eq:admissible_state_space} is positively
invariant under the networked SEIR dynamics \eqref{eq:seir}.
\end{proposition}

\begin{proof}
For each node $i\in\mathcal{V}$,
$\frac{d}{dt}(S_i+E_i+I_i+R_i)=0$, so the local normalization is preserved.
It remains to verify nonnegativity. If $S_i=0$, then
$\dot S_i=-\beta\sum_{j=1}^{n}a_{ij}S_iI_j=0$. If $E_i=0$, then
$\dot E_i=\beta\sum_{j=1}^{n}a_{ij}S_iI_j\geq0$. If $I_i=0$, then
$\dot I_i=\sigma E_i\geq0$, and if $R_i=0$, then
$\dot R_i=\gamma I_i\geq0$. Thus, the vector field is inward-pointing or
tangent on the boundary of the nonnegative orthant, and every solution
initialized in $\Omega$ remains in $\Omega$.
\end{proof}

Consequently, every compartmental variable lies in $[0,1]$ along trajectories
initialized in $\Omega$. All monomials in the graph-induced dictionaries and
all finite-order residuals are therefore uniformly bounded on this compact
set.

\subsection{Homogeneous block structure and dimensional reduction}

Proposition~\ref{prop:exact_degree_coupled_hierarchy} separates the two sources
of coupling in the lifted hierarchy. The matrices
$\matK_{k,k}^{\mathcal{E}}$ represent propagation within degree $k$ under the
linear compartmental field, whereas $\matK_{k,k+1}^{\mathcal{E}}$ represent
degree-raising interactions generated by the quadratic transmission field.
Linear saturation changes compartment labels from infectious to exposed at fixed
node indices but introduces neither new node indices nor new transmission
channels. The degree-raising blocks, by contrast, extend retained patterns only
through incoming effective transmission channels. Thus, the infinite lifted
operator has the exact block-bidiagonal structure described in
\eqref{eq:exact_homogeneous_block_dynamics}, while its sparsity remains tied to
the effective transmission organization of the original SEIR model.

The dimensional advantage is already apparent for the first edge-closure lifting. The
complete ordered quadratic tensor $x\otimes x$ has coordinate dimension
$(4n)^2=16n^2$, whereas the graph-induced quadratic dictionary in
\eqref{eq:phi2_edge_dictionary} has exactly $2m$ coordinates. Hence
\begin{equation}
\dim z_2^{\mathcal{E}}(x)=4n+2m.
\label{eq:edge_lift_dimension}
\end{equation}
Here $m=|\mathcal{E}|$ counts all effective transmission channels. This count
includes diagonal channels when within-subpopulation transmission is present.

\begin{proposition}[Dimensional reduction in sparse networks]
\label{prop:dimensional_reduction_sparse}
Assume that the effective transmission support satisfies $m=O(n)$. Then
$\dim z_2^{\mathcal{E}}(x)=O(n)$, whereas the complete ordered quadratic
Kronecker lifting has dimension $4n+(4n)^2=O(n^2)$. Thus, the first
edge-closure lifting reduces quadratic growth to linear growth in the number
of nodes. Adding at most one diagonal transmission channel per node preserves
this conclusion because it increases $m$ by at most $n$.
\end{proposition}

\begin{proof}
The result follows directly from \eqref{eq:edge_lift_dimension}. If $m=O(n)$,
then $4n+2m=O(n)$, while $4n+(4n)^2=O(n^2)$.
\end{proof}

For higher orders, total channel sparsity alone is insufficient because the
number of local transmission extensions depends on incoming connectivity. Let
\begin{equation*}
\mathcal{N}^{\mathrm{in}}(i)
=
\{j\in\mathcal{V}:(i,j)\in\mathcal{E}\},
\quad
\Delta^{\mathrm{in}}
=
\max_{i\in\mathcal{V}}|\mathcal{N}^{\mathrm{in}}(i)|.
\end{equation*}
The set $\mathcal{N}^{\mathrm{in}}(i)$ contains the epidemiological source
nodes whose infectious compartments contribute to the force of infection at
receiving node $i$. If $(i,i)\in\mathcal{E}$, then
$i\in\mathcal{N}^{\mathrm{in}}(i)$, so a diagonal channel contributes one to
the in-degree. Adding at most one diagonal channel per node increases each
in-degree by at most one and therefore preserves a uniform bounded-in-degree
hypothesis.

\begin{proposition}[Linear growth of linearly closed fixed-order liftings]
\label{prop:higher_order_linear_growth}
Consider a family of effective transmission networks satisfying
$\Delta^{\mathrm{in}}\leq\Delta$, where $\Delta\geq1$ is independent of $n$.
Let $d\geq2$ be fixed, and let $\mathcal{M}_k^{\mathcal{E}}$ be the distinct
commutative degree-$k$ dictionaries generated by the quadratic-generation and linear-
saturation recursion \eqref{eq:graph_induced_monomial_recursion}. Then, for every
fixed $k=2,\ldots,d$, there is a finite constant $C_k(\Delta)$ independent of
$n$ such that
\begin{equation}
|\mathcal{M}_k^{\mathcal{E}}|
\leq
C_k(\Delta)n.
\label{eq:fixed_degree_monomial_bound}
\end{equation}
One admissible recursive choice is
\begin{equation}
C_2(\Delta)=2\Delta,
\quad
C_{k+1}(\Delta)
=
2^{k+1}k\Delta\,C_k(\Delta),
\quad k\geq2.
\label{eq:fixed_degree_constant_recursion}
\end{equation}
Consequently, for fixed $d$,
\begin{equation}
\dim z_d^{\mathcal{E}}(x)
=
4n+\sum_{k=2}^{d}|\mathcal{M}_k^{\mathcal{E}}|
=
O(n).
\label{eq:fixed_order_linear_dimension}
\end{equation}
\end{proposition}

\begin{proof}
For $k=2$, \eqref{eq:graph_induced_monomial_base} gives
$|\mathcal{M}_2^{\mathcal{E}}|=2m$. Every channel contributes to exactly one
receiving-node in-neighborhood, including a diagonal channel when present, so
\begin{equation*}
m
=
\sum_{i\in\mathcal{V}}|\mathcal{N}^{\mathrm{in}}(i)|
\leq
n\Delta.
\end{equation*}
Therefore, $|\mathcal{M}_2^{\mathcal{E}}|\leq2\Delta n$, and
$C_2(\Delta)=2\Delta$ is valid.

Assume now that
$|\mathcal{M}_k^{\mathcal{E}}|\leq C_k(\Delta)n$ for some fixed $k\geq2$.
Take a commutative monomial
$q=x_{\alpha_1}\cdots x_{\alpha_k}\in\mathcal{M}_k^{\mathcal{E}}$, with
repeated factors allowed. By the product rule and
\eqref{eq:lie_derivative_definition},
\begin{equation*}
\mathcal{L}_{f_{\mathrm Q}}q
=
\sum_{s=1}^{k}
\left(\prod_{r\neq s}x_{\alpha_r}\right)
(f_{\mathrm Q})_{\alpha_s}(x).
\end{equation*}
Only susceptible and exposed components of $f_{\mathrm Q}$ are nonzero, and
each such component is a sum over at most $\Delta$ incoming transmission
channels. Differentiating one factor therefore generates at most $\Delta$
degree-$(k+1)$ candidates, while infectious and removed factors generate none.
Since $q$ has $k$ factor occurrences,
\begin{equation*}
\left|
\operatorname{Mon}_{k+1}
\left(\mathcal{L}_{f_{\mathrm Q}}q\right)
\right|
\leq
k\Delta.
\end{equation*}
By the definition of $\Gamma_{f_{\mathrm Q}}^{(k+1)}$, taking the union
over $q\in\mathcal{M}_k^{\mathcal{E}}$ and merging commutatively identical
products yields the quadratic seed bound
\begin{equation*}
|\mathcal{G}_{k+1}^{\mathcal{E}}|
\leq
k\Delta|\mathcal{M}_k^{\mathcal{E}}|.
\end{equation*}
The final degree-$(k+1)$ dictionary is
$\operatorname{Sat}_{f_{\mathrm L}}^{(k+1)}
(\mathcal{G}_{k+1}^{\mathcal{E}})$. Lemma~\ref{lem:finite_linear_closure},
applied at degree $k+1$, gives
\begin{equation*}
|\mathcal{M}_{k+1}^{\mathcal{E}}|
\leq
2^{k+1}|\mathcal{G}_{k+1}^{\mathcal{E}}|
\leq
2^{k+1}k\Delta C_k(\Delta)n.
\end{equation*}
Thus \eqref{eq:fixed_degree_constant_recursion} completes the induction.

The two multiplicative effects have different meanings. The factor $k\Delta$
counts degree-raising transmission extensions, whereas the factor $2^{k+1}$
is a uniform bound on same-degree infectious-to-exposed variants introduced by
linear saturation. Neither depends on $n$ when $k$ and $\Delta$ are fixed. A
diagonal channel may reproduce an existing node index and create repeated
factors, but it cannot exceed either bound. Finally, summing
\eqref{eq:fixed_degree_monomial_bound} over the fixed set of degrees
$2,\ldots,d$ gives \eqref{eq:fixed_order_linear_dimension}.
\end{proof}

The constants in \eqref{eq:fixed_degree_constant_recursion} are deliberately
conservative. Commutative canonicalization, repeated indices, overlapping local
patterns, and stabilization of the linear-saturation process can substantially reduce the
actual dictionaries. 
The proposition establishes the fixed-order scaling with network size, but the
displayed constants need not be tight.

\begin{corollary}[Comparison with complete ordered and symmetric liftings]
\label{cor:comparison_complete_polynomial_lifting}
Under the hypotheses of Proposition~\ref{prop:higher_order_linear_growth}, the
graph-induced state satisfies $\dim z_d^{\mathcal{E}}(x)=O(n)$ for fixed $d$.
The complete ordered Kronecker lifting up to order $d$ has dimension
\begin{equation}
D_{\mathrm{ord},d}
=
\sum_{k=1}^{d}(4n)^k
=
O(n^d),
\label{eq:complete_ordered_lifting_dimension}
\end{equation}
whereas the complete symmetric monomial lifting has dimension
\begin{equation}
D_{\mathrm{sym},d}
=
\sum_{k=1}^{d}\binom{4n+k-1}{k}
=
O(n^d).
\label{eq:complete_symmetric_lifting_dimension}
\end{equation}
Thus, both complete coordinate conventions have order-dependent polynomial
growth, while the linearly closed graph-induced lifting has linear growth
under the stated bounded-in-degree hypothesis.
\end{corollary}

\begin{proof}
Proposition~\ref{prop:higher_order_linear_growth} establishes the
$O(n)$ growth of the graph-induced dictionary at fixed lifting order.
In contrast, the complete ordered degree-$k$ tensor has $(4n)^k$
coordinates, yielding
\eqref{eq:complete_ordered_lifting_dimension}, whereas the complete
symmetric degree-$k$ dictionary has dimension
\eqref{eq:symmetric_degree_k_dimension}, which is $O(n^k)$ for fixed
$k$. Summing over the retained degrees yields
\eqref{eq:complete_symmetric_lifting_dimension}. Hence, despite their
finite-dimensional differences, the ordered and symmetric constructions both
scale as $O(n^d)$ at fixed maximum order $d$.
\end{proof}

\begin{remark}[Sparse channel count versus bounded in-degree]
A sparse total channel count is sufficient for linear growth of the first
edge-closure state, but not for the fixed-order conclusion in
Proposition~\ref{prop:higher_order_linear_growth}. A reciprocal directed star
without diagonal channels has $m=2(n-1)$ and hub in-degree $n-1$. Adding one
diagonal channel at every node gives $m=3n-2$ and hub in-degree $n$. In both
cases, \eqref{eq:edge_lift_dimension} remains linear in $n$, but no degree
bound independent of $n$ exists. The quadratic seed generation around the hub
may therefore produce a growing number of local extensions, and linear saturation
cannot restore a uniform bound. This example shows that sparsity in total
channel count does not by itself justify fixed-order linear scaling. It does
not assert a universal asymptotic law for all hub-dominated families.
\end{remark}

The distinction is epidemiologically relevant. Geographically structured or
regional transmission networks may have bounded incoming connectivity, whereas
hub-dominated mobility systems may not. The graph-induced construction makes
both mechanisms explicit: the total number of channels controls the first
edge-closure lifting, the quadratic field generates higher-order local patterns, and the
linear field completes the compartmental variants of those patterns without
introducing new nodes.

\subsection{Residual bounds for the first edge-closure lifting}

Recall that $z_2^{\mathcal{E}}(x)$ satisfies
\eqref{eq:lifted_dynamics_tensor}, with residual blocks defined in
\eqref{eq:residual_block_structure}--\eqref{eq:first_residual_se}. To avoid
confusion with the lifting order $d$, define the weighted in-degree
\begin{equation*}
\delta_i^{\mathrm{in}}
=
\sum_{\ell:(i,\ell)\in\mathcal{E}}a_{i\ell},
\quad
\delta_{\max}^{\mathrm{in}}
=
\max_{i\in\mathcal{V}}\delta_i^{\mathrm{in}}.
\end{equation*}
The quantity $\delta_i^{\mathrm{in}}$ is the total incoming transmission
intensity at receiving node $i$. If $(i,i)\in\mathcal{E}$, then $a_{ii}$ is
included as the within-subpopulation contribution. Thus, $m$,
$\delta_i^{\mathrm{in}}$, and $\delta_{\max}^{\mathrm{in}}$ consistently
include diagonal and off-diagonal channels.

\begin{proposition}[Residual bound for the first edge-closure lifting]
\label{prop:first_residual_bound}
For every $x\in\Omega$, the cubic residual $r_2^{\mathcal{E}}(x)$ satisfies
\begin{equation}
\left\|r_2^{\mathcal{E}}(x)\right\|_{\infty}
\leq
2\beta\delta_{\max}^{\mathrm{in}},
\label{eq:r2_infty_bound}
\end{equation}
and
\begin{equation}
\left\|r_2^{\mathcal{E}}(x)\right\|_2
\leq
\sqrt{5m}\,\beta\delta_{\max}^{\mathrm{in}}.
\label{eq:r2_two_bound}
\end{equation}
\end{proposition}

\begin{proof}
Because $x\in\Omega$, all compartmental variables belong to $[0,1]$. From
\eqref{eq:first_residual_si},
\begin{align*}
\left|\left(r_{SI}^{\mathcal{E}}(x)\right)_e\right|
&\leq
\beta\sum_{\ell:(i_e,\ell)\in\mathcal{E}}a_{i_e\ell}
=
\beta\delta_{i_e}^{\mathrm{in}}
\leq
\beta\delta_{\max}^{\mathrm{in}}.
\end{align*}
Similarly, \eqref{eq:first_residual_se} and the triangle inequality give
\begin{align*}
\left|\left(r_{SE}^{\mathcal{E}}(x)\right)_e\right|
&\leq
\beta\left(
\sum_{\ell:(i_e,\ell)\in\mathcal{E}}a_{i_e\ell}
+
\sum_{\ell:(j_e,\ell)\in\mathcal{E}}a_{j_e\ell}
\right)\\
&=
\beta\left(\delta_{i_e}^{\mathrm{in}}+
\delta_{j_e}^{\mathrm{in}}\right)
\leq
2\beta\delta_{\max}^{\mathrm{in}}.
\end{align*}
The first block of $r_2^{\mathcal{E}}$ is zero, so
\eqref{eq:r2_infty_bound} follows. The $SI$ and $SE$ residual blocks each
have $m$ components. Therefore,
\begin{equation*}
\left\|r_2^{\mathcal{E}}(x)\right\|_2^2
\leq
m\left(\beta\delta_{\max}^{\mathrm{in}}\right)^2
+
m\left(2\beta\delta_{\max}^{\mathrm{in}}\right)^2
=
5m\left(\beta\delta_{\max}^{\mathrm{in}}\right)^2.
\end{equation*}
Taking square roots yields \eqref{eq:r2_two_bound}.
\end{proof}

The bound separates network structure from the realized epidemic state. Large
weighted in-degree may reflect strong within-subpopulation transmission,
multiple external sources, large channel weights, or a combination of these.
It can therefore increase the uniform worst-case bound, but does not determine
the residual attained along a particular trajectory. Repeated factors and
algebraic overlap associated with diagonal channels may make the bound
conservative without invalidating it.

\begin{remark}[Residuals in low-prevalence regimes]
\label{rem:low_prevalence_residuals}
Every monomial in $r_{SI}^{\mathcal{E}}$ contains two infectious factors,
whereas every monomial in $r_{SE}^{\mathcal{E}}$ contains at least one. The
uniform bounds in Proposition~\ref{prop:first_residual_bound} may therefore be
substantially conservative at low prevalence. In particular,
$r_2^{\mathcal{E}}(x)=0$ whenever all infectious compartments vanish, and hence
at every disease-free equilibrium.
\end{remark}

\subsection{Homogeneous higher-order truncation residuals}

Linear saturation includes all same-degree terms generated by the linear field. By
Proposition~\ref{prop:order_d_homogeneous_residual}, an order-$d$ truncation
satisfies \eqref{eq:higher_order_tensor_dynamics}, and its residual has the
last-block form \eqref{eq:order_d_residual_last_block}. Every nonzero component
is generated by $\mathcal{L}_{f_{\mathrm Q}}\Phi_d^{\mathcal{E}}$ and is
therefore homogeneous of degree $d+1$. The residual is therefore a genuine polynomial-degree truncation forcing, not
a mixture of omitted same-degree and higher-degree coordinates.

\begin{proposition}[Boundedness of higher-order residuals]
\label{prop:higher_order_residual_bound}
For every lifting order $d\geq2$, there is a finite constant
$\varepsilon_d\geq0$ such that
\begin{equation*}
\left\|r_d^{\mathcal{E}}(x)\right\|_2
\leq
\varepsilon_d,
\quad x\in\Omega.
\end{equation*}
Moreover, every nonzero component of $r_d^{\mathcal{E}}$ is homogeneous of
degree $d+1$ and contains at least one infectious factor.
\end{proposition}

\begin{proof}
For fixed $d$, Proposition~\ref{prop:order_d_homogeneous_residual} expresses
$r_d^{\mathcal{E}}$ as a finite polynomial map whose nonzero entries are
components of $\mathcal{L}_{f_{\mathrm Q}}\Phi_d^{\mathcal{E}}$. Homogeneity
of degree $d+1$ follows from Lemma~\ref{lem:degree_propagation}. Each nonzero
component of $f_{\mathrm Q}$ contains an infectious factor, so the same is true
of every residual monomial. Since $\Omega$ is compact and the residual is
continuous, its Euclidean norm attains a finite maximum on $\Omega$, which
defines $\varepsilon_d$.
\end{proof}

In particular, $r_d^{\mathcal{E}}(x)=0$ whenever all infectious compartments
vanish. Proposition~\ref{prop:higher_order_residual_bound} is an existence
result based on continuity and compactness. It does not imply that
$\varepsilon_d$ decreases with $d$, that the sequence of finite liftings
converges, or that higher order necessarily improves a finite-horizon
trajectory approximation. The explicit graph-dependent formula in
Proposition~\ref{prop:first_residual_bound} is specific to the first
edge-closure lifting. Sharper computable bounds at arbitrary order require
additional combinatorial
and coefficient analysis.

The order-$d$ representation can thus be viewed as a graph-structured
linear system driven by a bounded homogeneous perturbation:
\begin{equation*}
\dot z
=
\matK_d^{\mathcal{E}}z+r_d^{\mathcal{E}}(x),
\quad
\left\|r_d^{\mathcal{E}}(x)\right\|_2
\leq
\varepsilon_d.
\end{equation*}
This interpretation permits linear-systems analysis in the retained
coordinates while keeping the omitted degree-$(d+1)$ forcing explicit.

\subsection{Epidemiological interpretation}

The results separate structural complexity from residual magnitude. The first
edge-closure lifting is controlled by the number of effective channels, whereas
higher-order growth reflects the local transmission patterns generated by the
quadratic field and completed under the linear compartmental dynamics.
Uniformly bounded incoming connectivity therefore yields linear fixed-order
growth, with diagonal and off-diagonal channels treated algebraically in the
same way.

The truncation residual collects the degree-$(d+1)$ transmission-supported
interactions omitted from an order-$d$ lifting. For the first edge-closure lifting, its uniform
bound depends explicitly on transmission rate, channel count, and maximum
weighted in-degree, while the value attained along a trajectory also depends
on the epidemic state and weight distribution. Every residual monomial contains
an infectious factor and hence vanishes at zero infectious prevalence.
Section~\ref{sec:numerical_illustrations} examines these effects numerically.

\section{Numerical illustrations}
\label{sec:numerical_illustrations}

This illustrative, non-calibrated study addresses three questions: dimensional
reduction relative to complete ordered Kronecker liftings, the first
edge-closure residual along SEIR trajectories, and the dependence of
higher-order dictionaries on local graph structure.

The residual
\begin{equation}
E_2(t)
=
\left\|r_2^{\mathcal{E}}(x(t))\right\|_2,
\label{eq:numerical_residual_measure}
\end{equation}
is evaluated along the original SEIR trajectory and therefore measures an
instantaneous closure defect, not a trajectory-prediction error. The
implementation computes the explicit $r_{SI}^{\mathcal{E}}$ and
$r_{SE}^{\mathcal{E}}$ blocks. The zero state block is omitted because it does
not affect $\|r_2^{\mathcal{E}}\|_2$. The topologies considered below are
illustrative rather than calibrated. They are chosen to isolate the effects of
sparsity, bounded in-degree, hub concentration, and the common
within-subpopulation transmission mechanism.

\subsection{Numerical setup, implementation, and reproducibility}

All simulations use the networked SEIR model in \eqref{eq:seir}, with
normalized local populations satisfying
$S_i(t)+E_i(t)+I_i(t)+R_i(t)=1$. The experiment is nondimensional: state
fractions, parameters, and time are interpreted in internally consistent model
units, with no conversion to SI units, calendar time, or a particular disease
scale. The illustrative parameter values are $\beta=0.8$, $\sigma=0.25$, and
$\gamma=0.1$. In the chosen model-time coordinate, the corresponding latent
and infectious time scales are $\sigma^{-1}=4$ and $\gamma^{-1}=10$.

The initial condition is localized at node $1$: $E_1(0)=I_1(0)=10^{-3}$,
$R_1(0)=0$, and $S_1(0)=1-E_1(0)-I_1(0)$. At all other nodes,
$S_i(0)=1$ and $E_i(0)=I_i(0)=R_i(0)=0$. The trajectories are evaluated on
$0\leq t\leq480$.

The reported structural quantities are the complete ordered quadratic
Kronecker dimension $D_{\mathrm{full},2}=4n+(4n)^2$ and the first edge-closure
dimension $D_{\mathrm{edge},2}=4n+2m$ from
\eqref{eq:edge_lift_dimension}. The dynamic quantity is the residual norm
$E_2(t)$ in \eqref{eq:numerical_residual_measure}. The right-hand side of
\eqref{eq:r2_two_bound} is denoted by
$B_2:=\sqrt{5m}\,\beta\delta_{\max}^{\mathrm{in}}$. Dimensional comparisons
use $n\in\{25,50,100,200,500\}$, whereas trajectory and residual calculations
use $n=100$.

Two complementary weighting configurations are used because they answer
different questions. In the unit-weight case, every active diagonal and
off-diagonal transmission channel has the same intensity. This convention is
natural for examining the effect of adding supported channels, but it also
allows the total incoming intensity to grow with in-degree. The second case
uses normalized weights, obtained by fixing the total incoming transmission
intensity at every receiving node:
\begin{equation}
\widetilde a_{ij}
=
\frac{a_{ij}}{\sum_{k=1}^{n}a_{ik}},
\quad
\sum_{j=1}^{n}\widetilde a_{ij}=1.
\label{eq:fixed_incoming_transmission_intensity}
\end{equation}
The normalization includes diagonal and off-diagonal channels and preserves the
support, so $m$, $D_{\mathrm{edge},2}$, and all graph-induced dictionary
dimensions are unchanged. Comparing the two cases distinguishes accumulation
of transmission intensity with in-degree from effects that persist through
support geometry, weight distribution, and nonlinear propagation.

Trajectories are computed with SciPy's \texttt{solve\_ivp}, using RK45 with
relative tolerance $10^{-8}$ and absolute tolerance $10^{-10}$. All summaries
are evaluated on the common grid $t_h=0.15h$, $h=0,\ldots,3200$. The grid peak
is $E_{2,\max}=\max_h E_2(t_h)$, with $t_{\mathrm{peak}}$ the first grid time
attaining it. The time average is
\begin{equation*}
\overline{E}_2
=
\frac{1}{T}\sum_{h=0}^{N-1}
\frac{E_2(t_h)+E_2(t_{h+1})}{2}(t_{h+1}-t_h),
\quad N=3200.
\end{equation*}

For the sparse random ensemble, master seed $12345$ defines deterministic
per-size streams $12345+n$, from which the $N_{\mathrm{rep}}=50$ realization
seeds are drawn and stored. Both weighting configurations use exactly the same
supports, seeds, initial condition, parameters, integration method, and output
grid. Ensemble summaries report arithmetic means and sample standard
deviations with denominator $N_{\mathrm{rep}}-1$, including the sample standard
deviation of the realization-specific peak times. The time-series band is the
pointwise sample standard deviation of the residual curves.

Automated validation checks dimensions, finiteness, local mass conservation,
and nonnegativity up to tolerance $10^{-8}$, without clipping. Across the
$104$ dynamic trajectories, the largest local conservation error is
$2.56\times10^{-15}$ and the minimum compartment value is
$-2.86\times10^{-9}$, within tolerance. Higher-order dictionaries are generated by applying
$\Gamma_{f_{\mathrm Q}}^{(k+1)}$ to produce the next-degree seed set and then
$\operatorname{Sat}_{f_{\mathrm L}}^{(k+1)}$ to obtain the saturated dictionary.
Compartment--node factors are canonically sorted, repeated factors are retained,
duplicate generation paths are merged, and saturation continues until no new
same-degree monomial appears.

\subsection{Network topologies}

Three representative network topologies are considered. They are a path graph,
a hub-dominated star graph, and a sparse random graph. The last topology lies
between regular bounded-in-degree connectivity and hub concentration.
Figure~\ref{fig:network_topologies} summarizes these three topology classes.

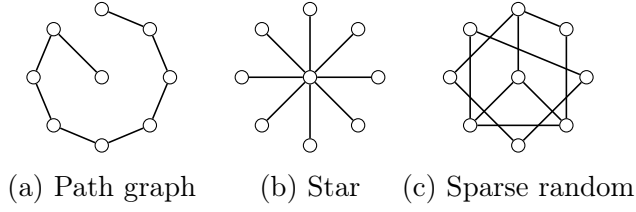
\begin{figure}[t]
\centering
\begingroup
\begin{tikzpicture}[
    node/.style={circle, draw, inner sep=1.4pt, minimum size=5pt},
    edge/.style={line width=0.6pt},
    panel/.style={font=\small},
    scale=0.95
]
\def\starbase#1{%
    \node[node] (#11) at (0,0.95) {};
    \node[node] (#12) at (0.67,0.67) {};
    \node[node] (#13) at (0.95,0) {};
    \node[node] (#14) at (0.67,-0.67) {};
    \node[node] (#15) at (0,-0.95) {};
    \node[node] (#16) at (-0.67,-0.67) {};
    \node[node] (#17) at (-0.95,0) {};
    \node[node] (#18) at (-0.67,0.67) {};
    \node[node] (#19) at (0,0) {};
}
\begin{scope}[xshift=0cm]
    \starbase{c}
    \foreach \i/\j in {1/2,2/3,3/4,4/5,5/6,6/7,7/8,8/9} {
        \draw[edge] (c\i) -- (c\j);
    }
    \node[panel] at (0,-1.55) {(a) Path graph};
\end{scope}
\begin{scope}[xshift=2.9cm]
    \starbase{s}
    \foreach \i in {1,...,8} {
        \draw[edge] (s9) -- (s\i);
    }
    \node[panel] at (0,-1.55) {(b) Star};
\end{scope}
\begin{scope}[xshift=5.8cm]
    \starbase{r}
    \draw[edge] (r1) -- (r2);
    \draw[edge] (r2) -- (r4);
    \draw[edge] (r4) -- (r6);
    \draw[edge] (r6) -- (r8);
    \draw[edge] (r8) -- (r3);
    \draw[edge] (r3) -- (r5);
    \draw[edge] (r5) -- (r7);
    \draw[edge] (r7) -- (r1);
    \draw[edge] (r1) -- (r9);
    \draw[edge] (r9) -- (r4);
    \draw[edge] (r9) -- (r6);
    \node[panel] at (0,-1.55) {(c) Sparse random};
\end{scope}
\end{tikzpicture}
\endgroup
\caption{Schematic representation of the three off-diagonal network topologies used in the numerical illustrations: path graph, star graph, and sparse random graph. The diagrams are illustrative and use fewer nodes than the numerical simulations. In the simulations, undirected contacts are represented by pairs of directed channels, and every node also has a diagonal within-subpopulation transmission channel. Diagonal loops are omitted here for visual clarity.}
\label{fig:network_topologies}
\end{figure}

The path graph connects each node to its immediate off-diagonal neighbors. Each undirected neighbor contact is represented by the two directed channels $(i,i+1)\in\mathcal{E}$ and $(i+1,i)\in\mathcal{E}$ for $i=1,\ldots,n-1$, and the $n$ diagonal channels $(i,i)\in\mathcal{E}$ are then included. Hence $m=2(n-1)+n=3n-2$, and the maximum in-degree is $3$ for $n\geq3$, so it remains uniformly bounded.

The star graph has one central node connected to every other node, with reciprocal off-diagonal channels $(1,j)\in\mathcal{E}$ and $(j,1)\in\mathcal{E}$ for $j=2,\ldots,n$, together with one diagonal channel at every node. Here again $m=3n-2$, but the maximum in-degree is $n$ because the receiving hub has $n-1$ off-diagonal source channels and one diagonal channel. The star is therefore sparse in total channel count but does not satisfy a uniform bounded-in-degree condition.

The sparse random graph is generated by independently sampling each off-diagonal ordered pair with probability $p=c/n$, with $c=4$, as in a directed Erd\H{o}s--R\'enyi construction \cite{ErdosRenyi1960}. The diagonal is not sampled: the $n$ within-subpopulation channels are added deterministically after the off-diagonal realization is generated. The expected total number of channels is therefore $n+c(n-1)=O(n)$, while the random off-diagonal component allows moderate local-connectivity heterogeneity.

For the sparse random topology, all reported quantities are averaged over
$N_{\mathrm{rep}}=50$ independent admissible realizations. The admission rule
is specified before simulation. Because the initial infection is localized at
node 1, that source node must have at least one off-diagonal outgoing
transmission channel. This rule excludes only realizations in which the initial
seed cannot reach any other subpopulation. Accepted seeds are retained and
reported for reproducibility. Tables report arithmetic means and sample
standard deviations. The grid peak time is summarized by its mean and sample
standard deviation. A grid-refinement check on the deterministic cases and
representative random realizations halves the output spacing from $0.15$ to
$0.075$. The largest relative change in the residual peak is below $0.24\%$.
The largest peak-time shift is one fine-grid step ($0.075$).

\begin{table}[t]
\centering
\caption{Unit-weight network quantities used to characterize the three supports before the fixed-intensity sensitivity is applied. For the sparse random graph, values are reported as mean $\pm$ sample standard deviation over $N_{\mathrm{rep}}=50$ independent realizations.}
\label{tab:network_quantities}
\begingroup
\small
\setlength{\tabcolsep}{4pt}
\begin{tabular}{lcccc}
\toprule
\textbf{Network} & $\boldsymbol{n}$ & $\boldsymbol{m}$ & $\boldsymbol{{\Delta}^{\mathrm{in}}}$ & $\boldsymbol{{\delta}_{\max}^{\mathrm{in}}}$ \\
\midrule
Path & 100 & 298 & 3 & 3 \\
Star & 100 & 298 & 100 & 100 \\
Sparse random & 100 & 499.76 $\pm$ 18.46 & 10.80 $\pm$ 1.25 & 10.80 $\pm$ 1.25 \\
\bottomrule
\end{tabular}
\endgroup
\end{table}

Table~\ref{tab:network_quantities} characterizes the common supports under unit weights. For $n=100$, the path and star both have $m=298$ effective channels, including $100$ diagonal and $198$ off-diagonal channels, but their maximum in-degrees are $3$ and $100$, respectively. The sparse random ensemble has $m=499.76\pm18.46$ and $\Delta^{\mathrm{in}}=10.80\pm1.25$. Under unit weights, maximum weighted in-degree equals maximum in-degree. Under \eqref{eq:fixed_incoming_transmission_intensity}, the support quantities $m$ and $\Delta^{\mathrm{in}}$ are unchanged, whereas $\delta_{\max}^{\mathrm{in}}=1$ for every realization and topology.

\subsection{Dimensional reduction}

The first numerical experiment compares the dimension of the proposed first edge-closure lifting with that of the complete ordered quadratic tensor lifting. By Proposition~\ref{prop:dimensional_reduction_sparse}, the ordered quadratic benchmark grows as $O(n^2)$, whereas the graph-induced lifting grows as $O(n)$ whenever $m=O(n)$.

For the path and star, both of which satisfy $m=3n-2$ after the diagonal channels are included, the first edge-closure dimension is $D_{\mathrm{edge},2}=4n+2(3n-2)=10n-4$, whereas the complete ordered quadratic tensor lifting has dimension $D_{\mathrm{full},2}=4n+16n^2$.

The reduction factors reported in Table~\ref{tab:dimension_comparison} use the
full ordered Kronecker tensor as the benchmark. A complete symmetric monomial
basis would remove permutation duplicates and yield smaller numerical
dimensions, although Corollary~\ref{cor:comparison_complete_polynomial_lifting}
shows that both complete benchmarks retain polynomial fixed-order growth in
$n$.

\begin{table}[t]
\centering
\caption{Comparison between complete ordered quadratic Kronecker tensor lifting dimensions and first edge-closure dimensions. For the sparse random graph, values are reported as mean $\pm$ standard deviation over $N_{\mathrm{rep}}=50$ independent realizations.}
\label{tab:dimension_comparison}
\begingroup
\small
\setlength{\tabcolsep}{4pt}
\begin{tabular}{lcccc}
\toprule
\textbf{Network} & $\boldsymbol{n}$ & $\boldsymbol{{D}_{\mathrm{full},2}}$ & $\boldsymbol{{D}_{\mathrm{edge},2}}$ & \textbf{Factor} \\
\midrule
Path & 100 & 160,400 & 996 & 161.04 \\
Star & 100 & 160,400 & 996 & 161.04 \\
Sparse random & 100 & 160,400 & 1,399.52 $\pm$ 36.91 & 114.69 $\pm$ 3.00 \\
\bottomrule
\end{tabular}
\endgroup
\end{table}

For $n=100$, the complete ordered quadratic tensor lifting has dimension
$D_{\mathrm{full},2}=160{,}400$, whereas the path and star both have
$D_{\mathrm{edge},2}=996$, corresponding to a reduction factor of approximately
$161.04$. Their equality is expected because both topologies have $m=3n-2$
effective transmission channels. Figure~\ref{fig:dimension_comparison} therefore
represents them with a single curve labeled ``Path/Star: edge''.

\begin{figure}[t]
\centering
\includegraphics[width=0.68\textwidth]{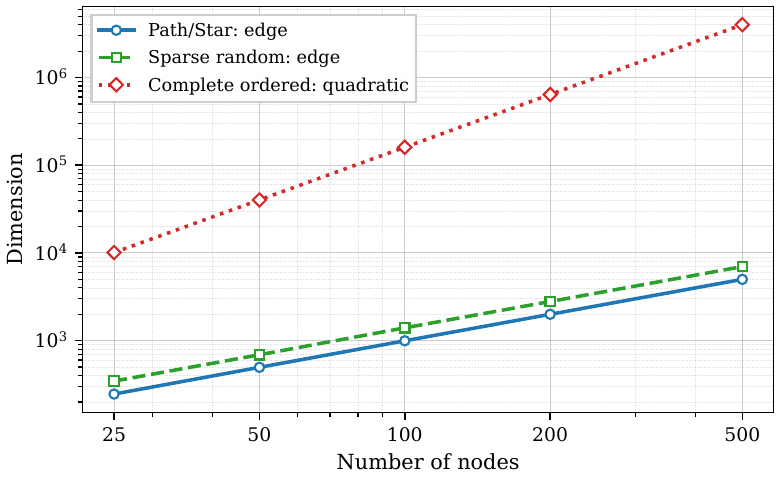}
\caption{Dimension of the complete ordered quadratic Kronecker tensor lifting and the first edge-closure lifting as a function of the number of nodes. The ordered quadratic benchmark grows as $O(n^2)$, while the graph-induced edge-closure dimension remains approximately linear over the sparse random graphs sampled here. The path and star curves coincide because both have $m=3n-2$ effective transmission channels after the common diagonal channels are included.}
\label{fig:dimension_comparison}
\end{figure}

The results agree with the predicted scaling. The complete ordered quadratic tensor lifting grows quadratically with $n$, while the graph-induced edge-closure dimension grows linearly for the path and star graphs and approximately linearly for the sparse random graph over the numerical range considered. Even for the random graph, the graph-induced dimension remains substantially
smaller than the ordered quadratic benchmark, and the gap widens with network
size because the dictionary is restricted to transmission-supported
observables.

\subsection{Residual dynamics and weighting sensitivity}

The second numerical experiment evaluates the first edge-closure residual
along trajectories of the original SEIR model. Since the lifted dynamics
satisfy \eqref{eq:lifted_dynamics_tensor}, $E_2(t)$ measures the pointwise
nonlinear forcing omitted from the finite dictionary. Table~\ref{tab:residual_summary}
reports the grid peak, time average, uniform bound, peak-to-bound ratio, and
peak time for both weighting configurations.

\begin{table}[H]
\centering
\caption{First edge-closure residuals for the two complementary weighting
configurations over $0\leq t\leq480$. Unit weights assign equal intensity to
every active transmission channel, whereas normalized weights satisfy
$\sum_j a_{ij}=1$ at every receiving node and thus impose a fixed total
incoming transmission intensity. For the sparse random graph, entries are
arithmetic mean $\pm$ sample standard deviation over
$N_{\mathrm{rep}}=50$ realizations. Deterministic values are reported without
a zero standard deviation. Time averages use the complete simulation horizon,
and peak times are expressed in nondimensional model-time units.}
\label{tab:residual_summary}
\begingroup
\small
\setlength{\tabcolsep}{2pt}
\begin{tabular}{@{}llccccc@{}}
\toprule
\textbf{Network} & \textbf{Weighting} & $\boldsymbol{{E}_{2,\max}}$ & $\boldsymbol{{\overline{E}}_2}$ & $\boldsymbol{{t}_{\mathrm{peak}}}$ & $\boldsymbol{{B}_2}$ & $\boldsymbol{{E}_{2,\max}/\mathbf{B}_2}$ \\
\midrule
Path & Unit & 0.25 & 0.085 & 15.60 & 92.64 & 0.0027 \\
 & Normalized & 0.12 & 0.080 & 30.90 & 30.88 & 0.0038 \\
\addlinespace[2pt]
Star & Unit & 3.81 & 0.029 & 4.65 & 3,088.04 & 0.0012 \\
 & Normalized & 0.61 & 0.018 & 22.05 & 30.88 & 0.0198 \\
\addlinespace[2pt]
Sparse & Unit & 1.87 $\pm$ 0.11 & 0.023 $\pm$ 0.001 & 11.00 $\pm$ 0.91 & 432.07 $\pm$ 52.84 & 0.0044 $\pm$ 0.0005 \\
 & Normalized & 0.77 $\pm$ 0.03 & 0.023 $\pm$ 0.001 & 35.16 $\pm$ 2.90 & 39.98 $\pm$ 0.74 & 0.0193 $\pm$ 0.0004 \\
\bottomrule
\end{tabular}
\endgroup
\end{table}

With unit weights, the star has the largest and earliest peak because its hub
combines many incoming channels with a much larger aggregate transmission
intensity. The path has the smallest peak but the largest time average because
a localized epidemic wave traverses the network over a much longer interval.
The sparse random ensemble is intermediate in peak magnitude, and the sample
standard deviation of its realization-specific peak time is reported
explicitly.

Fixing the total incoming transmission intensity changes the dynamic comparison
without changing any dictionary dimension. The star peak decreases from $3.81$
to $0.61$ and moves from $t=4.65$ to $t=22.05$, while the random-ensemble mean
peak becomes $0.77\pm0.03$. Thus, the unit-weight ranking is strongly influenced
by accumulation of transmission intensity at the hub. Nevertheless, support
geometry and distribution of the fixed intensity remain dynamically relevant.
The normalized path and star have the same channel count, maximum weighted
in-degree, and uniform bound $B_2=30.88$, yet their peak residuals differ by
more than a factor of five and occur at different times. The two cases more clearly separate support-dependent dictionary complexity
from weight-dependent residual forcing, without implying a universal
topology-only ordering.

The residual trajectories under unit and normalized channel weights are shown
in Figures~\ref{fig:residual_dynamics_unit} and
\ref{fig:residual_dynamics_fixed}, respectively.

The extended horizon resolves the apparent plateau previously observed for the
normalized path. At $t=240$ and $t=360$, its residual is still $0.1113$ and
$0.0868$, respectively, because a localized epidemic wave is propagating along
the path and the global residual norm changes little under translation of that
wave. After the wave reaches the terminal part of the network, the residual
decays to $3.32\times10^{-12}$ at $t=480$. Across all $104$ trajectories, the
largest final residual is below $1.90\times10^{-11}$ and the largest final value
of $E_i+I_i$ is below $1.47\times10^{-5}$, confirming that the simulations have
entered the final extinction phase.

\begin{figure}[H]
\centering
\includegraphics[width=0.72\textwidth]{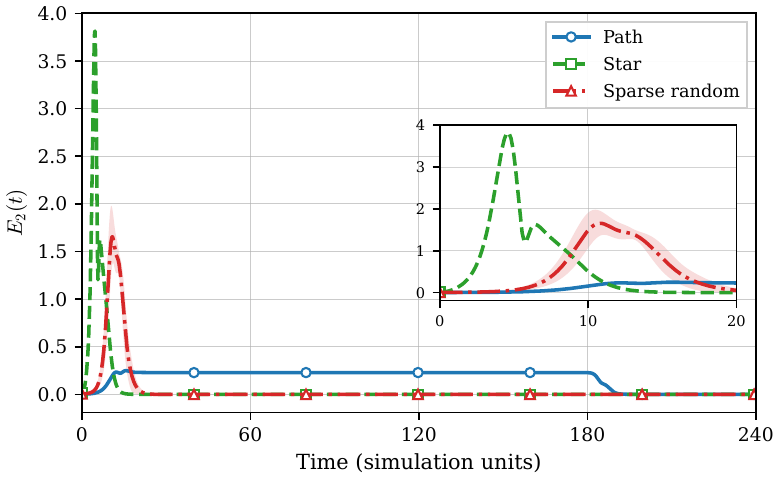}
\caption{Time evolution of $E_2(t)=\|r_2^{\mathcal{E}}(x(t))\|_2$ with unit channel weights over $0\leq t\leq240$. The inset over $0\leq t\leq20$ resolves the early peaks. Sparse-random curves are pointwise ensemble means, and shaded bands show one pointwise sample standard deviation.}
\label{fig:residual_dynamics_unit}
\end{figure}

\begin{figure}[H]
\centering
\includegraphics[width=0.72\textwidth]{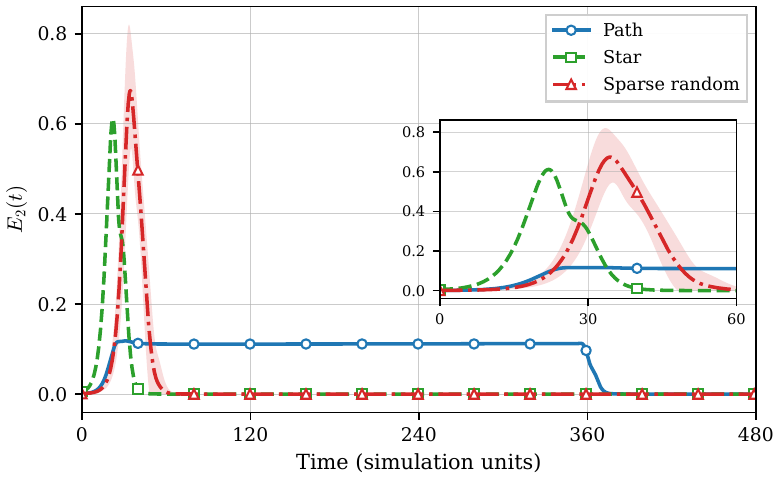}
\caption{Time evolution of $E_2(t)=\|r_2^{\mathcal{E}}(x(t))\|_2$ under
normalized weights, $\sum_j a_{ij}=1$, over $0\leq t\leq480$.
The inset over $0\leq t\leq60$ resolves the delayed initial peaks, while the
full horizon displays the final extinction phase. Sparse-random curves are
pointwise ensemble means, and shaded bands show one pointwise sample standard
deviation. The long path plateau is caused by a traveling epidemic wave rather
than a nonzero asymptotic residual.} \label{fig:residual_dynamics_fixed}
\end{figure}

Row normalization lowers and delays the star transient and slows propagation
along the path. All peaks remain below the uniform bounds in
Proposition~\ref{prop:first_residual_bound}, and the peak-to-bound ratios
illustrate their conservatism. The curves remain instantaneous closure defects,
not trajectory-error estimates.

\subsection{Higher-order dictionaries and local graph patterns}

The third experiment examines the saturated degree-three dictionary. The
quadratic field generates $\mathcal{G}_3^{\mathcal{E}}$ from
$\mathcal{M}_2^{\mathcal{E}}$, and linear saturation yields
$\mathcal{M}_3^{\mathcal{E}}$. Diagonal channels may generate repeated-factor
products such as $S_iI_i^2$, $S_iI_iE_i$, and $S_i^2I_i$. Products are treated
as commutative monomials, with permutations identified and duplicate paths
merged. The reported values therefore refer to the final saturated dictionary,
not only to its transmission-generated seeds.

For each topology, the computation records separately the number of cubic
seeds, the coordinates added by linear saturation, and the final value
\(\dim\Phi_3^{\mathcal{E}}=|\mathcal{M}_3^{\mathcal{E}}|\). Table~\ref{tab:cubic_dictionary_size}
reports the final dictionary and lifted-state dimensions. The complete ordered
Kronecker lifting dimension is evaluated from
\eqref{eq:complete_ordered_lifting_dimension} with \(d=3\), while the
graph-induced order-three state contains the original state together with the
closed degree-two and degree-three blocks.

\begin{table}[b]
\centering
\caption{Comparison between complete ordered Kronecker lifting dimensions up to
order three and linearly closed graph-induced cubic dictionary sizes. For the
sparse random graph, values are reported as mean \(\pm\) sample standard
deviation over \(N_{\mathrm{rep}}=50\) independent realizations.}
\label{tab:cubic_dictionary_size}
\begingroup
\small
\setlength{\tabcolsep}{3pt}
\begin{tabular}{lccccc}
\toprule
\textbf{Network} & \(\boldsymbol{n}\) &
\(\boldsymbol{D_{\mathrm{full},3}}\) &
\(\boldsymbol{\dim\Phi_2^{\mathcal{E}}}\) &
\(\boldsymbol{\dim\Phi_3^{\mathcal{E}}}\) &
\(\boldsymbol{\dim z_3^{\mathcal{E}}}\) \\
\midrule
Path & 100 & 64,160,400 & 596 & 3,462 & 4,458 \\
Star & 100 & 64,160,400 & 596 & 41,486 & 42,482 \\
Sparse random & 100 & 64,160,400 &
\(999.52 \pm 36.91\) &
\(11,214.84 \pm 816.19\) &
\(12,614.36 \pm 852.65\) \\
\bottomrule
\end{tabular}
\endgroup
\end{table}

For \(n=100\), the path produces \(2{,}176\) cubic seeds and gains
\(1{,}286\) coordinates through linear saturation, yielding
\(\dim\Phi_3^{\mathcal{E}}=3{,}462\). The star produces \(25{,}941\)
seeds and gains \(15{,}545\) saturation coordinates, yielding \(41{,}486\).
For the sparse random ensemble, the corresponding values are
\(7{,}046.76\pm516.53\) seeds,
\(4{,}168.08\pm300.35\) saturation additions, and a final dictionary size of
\(11{,}214.84\pm816.19\). The standard deviation of the final size is computed
directly across realization-specific closed dictionaries and is not obtained
by adding the two component standard deviations.

At $n=100$, all three graph-induced order-three states remain several orders
of magnitude smaller than the complete ordered dimension $64{,}160{,}400$.
Saturation enlarges each seed set but preserves the ordering: path, sparse
random ensemble, then hub-dominated star.

\begin{figure}[t]
\centering
\includegraphics[width=0.68\textwidth]{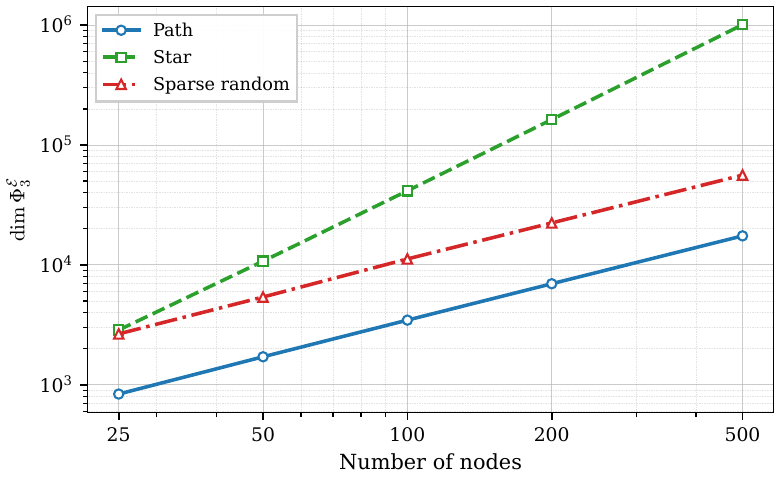}
\caption{Size of the saturated graph-induced cubic dictionary for the three
transmission topologies, including one diagonal channel per node. Values include
both quadratic-field seeds and same-degree coordinates added by linear
saturation.}
\label{fig:higher_order_dictionary_size}
\end{figure}

Figure~\ref{fig:higher_order_dictionary_size} illustrates the distinction
between sparsity and uniformly bounded maximum in-degree after linear saturation
has been enforced. The path retains bounded maximum in-degree and exhibits
linear cubic-dictionary growth. The star remains sparse in total channel count
but develops a much larger, approximately quadratic cubic dictionary because
the hub in-degree grows with \(n\). Adding one diagonal channel per node changes
each in-degree by only one and therefore does not alter this distinction. The
sparse random ensemble remains close to linear over the sampled range, although
this is an empirical finite-range observation rather than an asymptotic claim
for the random ensemble. The results are consistent with
Proposition~\ref{prop:higher_order_linear_growth}: linear saturation changes the
degree-dependent constants, but not the fixed-order \(O(n)\) conclusion for
families with uniformly bounded local connectivity.

\subsection{Discussion}

The experiments assess structural scalability and residual behavior rather
than provide a calibrated epidemic forecast. The first edge-closure dimension is governed by
channel count, while higher-order size reflects local patterns and their linear
saturation. Residual peaks remain well below the conservative uniform bounds
and show the influence of weighted in-degree, support geometry, and weight
distribution, but they establish neither a universal topology ranking nor an
epidemiological policy conclusion. Prediction, stability, and control remain
outside the scope of this numerical study.

\section{Conclusions}
\label{sec:conclusions}

This paper has developed a graph-induced tensor lifting for deterministic
networked SEIR models based on the complete effective transmission support. The framework yields an exact graph-supported quadratic representation and a first
edge-closure lifting built from susceptible--infectious and
susceptible--exposed observables. Diagonal and off-diagonal channels enter the
same algebraic construction, while retaining their local and cross-population interpretations. Higher degrees are generated by the
quadratic transmission field and saturated under the linear compartmental
dynamics. The resulting finite system is linear up to an explicit next-degree
residual.

The first lifted state has dimension $4n+2m$, including diagonal channels, and
therefore grows linearly when $m=O(n)$. At fixed order, the saturated
higher-degree dictionaries also grow linearly under a uniformly bounded maximum
in-degree, whereas complete ordered and symmetric polynomial liftings scale as
$O(n^d)$. Thus, bounded local connectivity rather than channel sparsity alone
is the relevant higher-order condition. Hub-dominated networks may have only $O(n)$
channels while producing much larger dictionaries.

The residual is the explicit next-degree transmission forcing omitted from the
finite dictionary. It is not a trajectory-prediction error. For the first
edge-closure lifting, its uniform bounds depend on transmission rate, channel
count, and maximum weighted
incoming intensity. Every residual monomial contains an infectious factor, so
the residual vanishes at disease-free equilibria and may be small at low
prevalence. At higher orders, continuity on the admissible compact state space
gives finite bounds, but convergence with increasing lifting order is not
established.

The numerical experiments separate support-dependent complexity from dynamic
weighting effects. Unit and normalized weights produce identical dictionary
dimensions but different residual trajectories. Normalization substantially
lowers and delays the star transient, while the path and star retain distinct
profiles even with equal channel count and total incoming intensity. The
higher-order experiment confirms the contrast between bounded local
connectivity and hub concentration: the path exhibits linear cubic growth,
whereas the star produces a much larger dictionary. These conclusions remain
specific to the selected parameters, initial condition, and graph families.

The framework also provides a basis for subsequent Lyapunov analysis,
residual-aware predictive control, sharper state-dependent bounds, and estimates
connecting closure residuals with finite-horizon trajectory errors. These
developments are left for future work. Extensions to other networked
compartmental systems are possible when nonlinear processes are supported on
identifiable interaction channels and differentiation yields a manageable
polynomial hierarchy. Nonpolynomial incidence, explicit migration, delays, and
additional nonlinear transitions require separate constructions. Together,
these results characterize a transmission-supported, linearly closed hierarchy,
its dimensional growth, and its explicit next-degree residual.

\section*{CRediT authorship contribution statement}

Enrique Baeyens: Conceptualization, Methodology, Formal analysis, Software,
Validation, Data curation, Visualization, Writing -- original draft, Writing
-- review and editing.

\section*{Funding}

This research did not receive any specific grant from funding agencies in the
public, commercial, or not-for-profit sectors.

\section*{Declaration of generative AI and AI-assisted technologies in the
manuscript preparation process}

During the preparation of this work, the author used OpenAI's ChatGPT and Codex
to support language editing and code review.
The author reviewed and edited all AI-assisted output and takes full
responsibility for the content of the manuscript.

\section*{Declaration of competing interest}

The author declares no known competing financial interests or personal
relationships that could have appeared to influence the work reported in this
paper.

\section*{Data and Code Availability}

The source code, synthetic outputs, tables, and figures are archived as version
1.0.2 on Zenodo~\cite{Baeyens2026Software}.

\noindent Archived release: \url{https://doi.org/10.5281/zenodo.21427810}.
No empirical, confidential, or third-party data were used in this work.

\bibliographystyle{elsarticle-num}
\setlength{\bibsep}{3pt}
\bibliography{references}

\end{document}